\documentclass[11pt]{article}
\usepackage{amsfonts}
\usepackage{amsmath,amsthm}
\usepackage{hyperref}
\usepackage{slashed}
\usepackage[all]{xy}
\usepackage{graphicx}

\usepackage{mathrsfs}
\usepackage{bbm}

\usepackage{amssymb}
\usepackage{latexsym}

\DeclareMathAlphabet{\mathpzc}{OT1}{pzc}{m}{it}
\DeclareMathAlphabet\EuFrak{U}{euf}{m}{n}	
\SetMathAlphabet\EuFrak{bold}{U}{euf}{b}{n}	

\newcommand{\wa}{\widehat}

\newcommand{\Aut}{ {\bf aut} }

\newcommand{\Div}{ {\bf div} }

\newcommand{\bs}{\boldsymbol}
\newcommand{\bsh}{{\boldsymbol h}}
\newcommand{\bsx}{{\boldsymbol x}}
\newcommand{\bsy}{{\boldsymbol y}}
\newcommand{\bsa}{{\boldsymbol a}}

\newcommand{\bsg}{{\boldsymbol g}}
\newcommand{\bsk}{{\boldsymbol k}}
\newcommand{\bsi}{{\boldsymbol i}}
\newcommand{\bsj}{{\boldsymbol j}}
\newcommand{\bsm}{{\boldsymbol m}}
\newcommand{\bsf}{{\boldsymbol f}}

\newcommand{\bsv}{{\boldsymbol v}}

\newcommand{\bmu}{{\boldsymbol \mu}}

\newcommand{\bsdelta}{{\boldsymbol \delta}}

\newcommand{\bsA}{{\boldsymbol A}}
\newcommand{\dbsA}{\dot{\bsA_{ \, }}}
\newcommand{\bsE}{{\boldsymbol E}}

\newcommand{\bI} {{\mathbb{I}}}

\newcommand{\bC} {{\mathbb C}}

\newcommand{\bR} {{\mathbb R}}

\newcommand{\bU} {{\mathbb U}}

\newcommand{\bSO} {{\mathbb{SO}}}

\newcommand{\bZ} {{\mathbb Z}}

\newcommand{\bN} {{\mathbb N}}

\newcommand{\iGamma}{{\it{\Gamma}}}

\newcommand{\mA}{\mathcal A}
\newcommand{\mB}{\mathcal B}
\newcommand{\mC}{\mathcal C}
\newcommand{\mD}{\mathcal D}

\newcommand{\mF}{\mathcal F}

\newcommand{\mH}{\mathcal H}

\newcommand{\mL}{\mathcal L}

\newcommand{\mS}{\mathcal S}

\newcommand{\mU}{\mathcal U}

\newcommand{\mW}{\mathcal W}

\newcommand{\efh}{\EuFrak h}

\newcommand{\efB}{\EuFrak{B}}

\newcommand{\efF}{\EuFrak{F}}

\newcommand{\efH}{\EuFrak{H}}

\newcommand{\sG}{\mathscr{G}}

\newcommand{\lb}{{\left\langle \right.}}
\newcommand{\rb}{{\left. \right\rangle}}

\newcommand{\Dom}{ \, {\mathrm{Dom}} \, }

\newcommand{\supp}{{\mathrm{supp}} \, }

\newtheorem{thm}{Theorem}[section]

\newtheorem{lem}[thm]{Lemma}
\newtheorem{prop}[thm]{Proposition}

\theoremstyle{definition}
\newtheorem{ex}{Example}[section]

\theoremstyle{remark}

\numberwithin{equation}{section}

\begin{document}

\author{\textsc{Ezio Vasselli} \\
\small{Dipartimento di Matematica, Universit\`a di Roma ``Tor Vergata'',}\\
\small{Via della Ricerca Scientifica, I-00133 Roma,  Italy.}  \\
\small{\texttt{ ezio.vasselli@gmail.com }}\\[20pt]
}

\title{Twisting factors\\and\\fixed-time models in quantum field theory}
\maketitle

\begin{abstract} 
We construct a class of fixed-time models in which the commutations relations 
of a Dirac field with a bosonic field are non-trivial and 
depend on the choice of a given distribution ("twisting factor"). 
If the twisting factor is fundamental solution of a differential operator, 
then applying the differential operator to the bosonic field yields
a generator of the local gauge transformations of the Dirac field.
Charged vectors generated by the Dirac field define states
of the bosonic field which in general are not local
excitations of the given reference state.
The Hamiltonian density of the bosonic field presents a non-trivial interaction term:
besides creating and annihilating bosons, it acts on momenta of fermionic wave functions.
When the twisting factor is the Coulomb potential,
the bosonic field contributes to the divergence of an electric field
and its Laplacian generates local gauge transformations of the Dirac field.
In this way we get a fixed-time model fulfilling the equal time commutation relations
of the interacting Coulomb gauge. 

\medskip


\end{abstract}


\section{Introduction}
\label{sec.intro}

A typical class of models in quantum field theory is given by those formed 
by a Dirac (electron) field $\psi$ and a bosonic field, 
say $\varphi^\lambda$ for future convenience, 
fulfilling the equal time commutation relations
\begin{equation}
\label{eq.intro.2}
[ \varphi^\lambda(\bsx,t) , \psi(\bsy,t) ] 
\ = \ 
- \sigma(\bsy - \bsx) \psi(\bsy,t) 
\, .
\end{equation}
In the previous expression
$\sigma$ is a distribution that evidently measures the obstruction for $\psi$ and 
$\varphi^\lambda$ to commute. Moreover, in dependence on the support of $\sigma$, 
it also yields an obstruction for them being relatively local.
Just to mention two examples, we encounter commutation relations of the type (\ref{eq.intro.2})
in the derivative coupling model \cite[\S 4.6.3]{Strocchi} (where $\sigma$ is the Dirac delta),
and in the Coulomb gauge (QED), with $\sigma$ the Coulomb potential
\cite[\S 7.8.1]{Strocchi}.

\medskip 

Now, a way to produce models where (\ref{eq.intro.2}) holds is to start from commuting free fields
$\psi$ and $\varphi$ and then perform a Wick product of the type $:e^{-i\varphi(x)}:\psi(x)$,
as in the derivative coupling model.
Nevertheless this method does not allow to escape from the Borchers class of the initial model,
so we do not really gain additional structural informations.
A further point is that when (\ref{eq.intro.2}) involves non-local field systems,
as in the Coulomb gauge, important tools of local Wightman fields are not available, 
or at least problematic to handle (spin-statistic theorem, perturbative series, renormalization).

\medskip 

In the present paper we propose a different approach to obtain commutation relations 
of the type (\ref{eq.intro.2}), based on operator algebraic techniques. 
Starting from a fixed time field system $\mF = \{ \psi , \varphi \}$
formed by commuting free fields, 
we construct a new field system $\mF^\lambda = \{ \psi , \varphi^\lambda \}$.
The field $\varphi^\lambda$, fulfilling (\ref{eq.intro.2}), does not act merely on the bosonic Fock space,
and $\mF^\lambda$ is not relatively local to $\mF$ when $\sigma$ has support with non-empty interior.
In this way typically we obtain non-local field systems, thus the task of 
exhibiting local observables encoding the informations necessary to reconstruct
$\mF^\lambda$ arises. 
We will approach this problem in a relativistic setting 
by making use of current and charge density fields \cite{Vas3}, 
with the idea of taking as model the interacting Coulomb gauge whose 
equal time commutation relations are easily realized by one of our models 
\S \ref{sec.syn}. 
In the present paper we make a preliminary step towards this direction,
by studying the charged states of $\varphi^\lambda$ induced by states of $\psi$.

\medskip 

In our hypothesis, we may assume that $\varphi$ is defined by a regular representation
of the Weyl algebra $\mW$ associated with the symplectic space $\mS$ of test functions 
of the type $s = (s_0,s_1) \in \mS(\bR^3) \oplus \mS(\bR^3)$ and that,
analogously, $\psi$ arises from a representation of the CAR algebra
$\mC_\bsh$ where 
$\bsh \doteq L^2(\bR^3,\hat{\bC}^4)$, $\hat{\bC}^4 \doteq \bC^4 \oplus \bC^{4,*}$,
is the Hilbert space of Dirac (electron-positron) spinors at a fixed time.
Starting from these C*-algebraic structures, 
our main ingredient is the distribution $\sigma \in \mS'(\bR^3)$
that we use to construct a morphism
\[
u_\sigma : \mS \to \mU(\bsh) \, ,
\]
where $\mS$ is intended as an additive group. Such an action is defined
by means of the exponential of the convolution $\sigma \star s_0 \in C^\infty(\bR^3)$,
and allows to construct what we call a \emph{fermionic representation} of the Weyl algebra,
with "twisted Weyl unitaries" $W^\lambda(s)$ acting non-trivially on the fermionic Hilbert space.
The adjoint action of the Weyl unitaries induces automorphisms of the CAR algebra,
given by
\begin{equation}
\label{eq.intro.1}
W^\lambda(s) \, \psi(w) \, W^\lambda(-s)  
\ = \ 
\psi(e^{-i \sigma \star s_0}w)
\, ,
\end{equation}
where $s \in \mS$ and $w \in \bsh_+ \doteq L^2(\bR^3,\bC^4)$
{\footnote{
The relations (\ref{eq.intro.1}) are an example of semi-direct product algebra
in the sense of \cite{Her96,Her98}. In the above references, a Weyl algebra 
relative to an asymptotic electromagnetic field is considered. 
Here we are interested to different Weyl algebras, at a fixed time and with 
a different interpretation.
}}.
Provided that the fermionic representation of the Weyl algebra is regular, 
we get the bosonic field $\varphi^\lambda$ fulfilling the commutation relations
(\ref{eq.intro.2}), that are therefore interpreted as the infinitesimal version of (\ref{eq.intro.1}).

\medskip 

The techniques used to construct $\varphi^\lambda$ are not new, 
in fact they simply involve the fermionic 2nd quantization of
a selfadjoint operator on $\bsh$ defined by the twisting factor.
Yet the so-obtained class of models exhibits properties that,
at author's knowledge, have not been previously remarked. 

First, the Weyl algebra acquires a family of states
labelled by fermionic vectors. The corresponding GNS representations are given 
by the action of the twisted Weyl operators on the cyclic subspaces generated by vectors of the
type $\Omega^q_f \otimes \Omega_b$, where $\Omega^q_f$ is a vector in the 
fermionic Fock space of $\bsh$ with charge $q \in \bZ$ and $\Omega_b$ is the reference state of $\varphi$.
Computing the expectation values of $\varphi^\lambda$ with respect to one of the above
states, we obtain -- for different $m \in \bN$ -- 
non-trivial transition amplitudes between the Hilbert spaces
$\mH^\lambda_m$ of $m$-fold excitations of $\Omega^q_f \otimes \Omega_b$, 
in contrast to the untwisted field $\varphi$.
A further analysis of the expectation values shows that our states in general 
cannot be interpreted in terms of automorphisms of the Weyl algebra
induced by external (c-number) potentials.

Secondly, provided that $\sigma$ is fundamental solution of 
a differential operator $P_\partial$, defining the field 
$\varrho \doteq P_\partial \varphi^\lambda$ we get the commutation relations
\[
[ \varrho(s) \, , \, \psi(w) ] \ = \ - \psi(s_0 w) \, .
\]
Thus $\varrho(s)$ is a generator of the local gauge transformations
of the Dirac field,
and behaves like the electron-positron charge for all spinors $w$ 
such that $s_0 \restriction \supp(w) = 1$.
The field $\varrho$ is relatively local to $\psi$
independently of the support of $\sigma$, and the vector states defined by
$\Omega^q_f \otimes \Omega_b$, when restricted to the subalgebra generated by
$W(P_\partial s) \in \mW$, $s \in \mS$, are localized on the support of $\Omega^q_f$ (regarded as a function with values in some antisymmetric 
tensor power of $\hat{\bC}^4$).

\medskip 

These results are proved in the following \S \ref{sec.A} 
and resumed in our main Theorem \ref{sec.syn}.
We provide examples for $\sigma$ the Lebesgue measure, the Yukawa potential
and, finally, the Coulomb potential.
In the latter case we construct an electric field and get the
equal time commutation relations 
of the interacting Coulomb gauge, with $\bs\Delta \varphi^\lambda$ contributing to 
the charge density field; we stress that in this case the observables are manifest,
being given by the electric field and gauge-invariant combinations of the Dirac field (e.g. the free Dirac Hamiltonian \cite[\S 2]{Bon70}).

\medskip 

In the final \S \ref{sec.H} we perform an analysis in momentum space and give 
an expression for the Hamiltonian of $\varphi^\lambda$, that presents 
interaction terms acting both on the bosonic and fermionic factor of the Fock 
space: in particular, it induces an exchange of momenta between fermionic and 
bosonic states. 
This is useful to get a hint of the form that our models may take in the
relativistic setting. The analysis of this latter scenario, as well as the 
necessary discussion of Einstein causality for the observables, 
are postponed to \cite{Vas3}.

\medskip

Finally we mention that C*-algebraic aspects of our construction are investigated
in the companion paper \cite{Vas1}. There, a C*-algebra $\bs\mF$ generated by the symbols
$\psi(w)$ and elements of a given C*-algebra $\mA$ is defined,
with relations generalizing (\ref{eq.intro.1}).
The construction involves creation and annihilation operators on a suitable
fermionic Fock bimodule over $\mA$.
When $\mA$ is a Weyl algebra, $\bs\mF$ is a semidirect product algebra 
of the type \cite{Her96,Her98}, and the fields constructed in the present paper define
regular representations of $\bs\mF$.

\section{The models}
\label{sec.A}

In the present section we present our family of static models, firstly by considering 
generic states of the Weyl algebra and then passing to regular states (hence to bosonic fields).
In GNS representations, the main point is that the Weyl algebra acts 
non-trivially on the fermionic factor of the Hilbert space,
and as a consequence it gains a family of states labeled by the electron-positron charge. 
This point shall be evident in the Examples \S \ref{sec.syn}.

%

\paragraph{Stone generator of the twisting factor.}
We start by considering the symplectic space $\mS \doteq \mS(\bR^3) \oplus \mS(\bR^3)$ 
of compactly supported, real test functions with symplectic form 
\begin{equation}
\label{eq.00.i1}
\eta(s,s') \, \doteq \, \int ( s_1 s'_0 - s_0 s'_1 ) 
\ \ \ , \ \ \ 
s=(s_0,s_1),s' \in \mS
\, .
\end{equation}
We then define the Hilbert spaces
$\bsh_+ \doteq L^2(\bR^3,\bC^4)$, $\bsh_- \doteq L^2(\bR^3,\bC^{4,*})$
and $\bsh \doteq \bsh_+ \oplus \bsh_-$ endowed with the standard scalar product 
\cite[Eq.13]{Bon70}. On $\bsh$ we define the conjugation
\begin{equation}
\label{eq.00.i2}
\kappa (w_+ \oplus \bar{w}_-) \, \doteq \, w_- \oplus \bar{w}_+
\ \ \ , \ \ \ 
w_+ , w_- \in \bsh_+ \, .
\end{equation}
In the previous expression, $\bar{w}_-$ is the $\bC^{4,*}$-valued map 
defined by the conjugate of $w_-$. 
We interpret $\bsh_+$ as the particle (electron) space, and 
$\bsh_-$ as the antiparticle (positron) space. 
Given the tempered distribution $\sigma \in \mS'(\bR^3)$ 
we have $\sigma \star s_0 \in C^\infty(\bR^3)$ for all $s \in \mS$, 
and we define the unitary-valued mapping 
$u_{\sigma,s} \in C^\infty(\bR^3,\bU(\hat{\bC}^4))$,
\begin{equation}
\label{eq.lem.A.1'}
u_{\sigma,s}(\bsx) \, \doteq \, 
\left(
\begin{array}{cc}
e^{-i\sigma \star s_0(\bsx)} {\bf 1}_4 \bf & 0 \\
0 &  e^{i\sigma \star s_0(\bsx)} {\bf 1}_4
\end{array}
\right)
\, ,
\end{equation}
where ${\bf 1}_4$ is the $4 \times 4$ identity matrix.
Applying $u_{\sigma,s}$ to spinors $w \in \bsh$ we obtain a unitary operator,
and this yields the group morphism
$u_\sigma : \mS \to \mU(\bsh)$, $s \mapsto u_{\sigma,s}$,
that we call the \emph{twist}. 
$u_\sigma$ is regular: given $s \in \mS$ and $w \in \bsh$,
we have
\begin{equation}
\label{eq.lem.A.1}
\lim_{t \to 0} 
\int | u_{\sigma,ts}(\bsx) w(\bsx) - w(\bsx) |^2 d^3\bsx 
\, = \, 
0 \ ,
\end{equation}
thus we may write $u_{\sigma,s} = e^{i\bmu_{\sigma,s}}$
where the Stone generator $\bmu_{\sigma,s}$ can be written explicitly   
\begin{equation}
\label{eq.00.08}
\bmu_{\sigma,s}w
\, = \, 
(- \sigma \star s_0 \cdot w_+) \oplus ( \sigma \star s_0 \cdot \bar{w}_- )
\end{equation}
on Dirac spinors $w$ such that $\sigma \star s_0 \cdot w \in L^2(\bR^3,\hat{\bC}^4)$.
Spinors $w \in C^\infty_c(\bR^3,\hat{\bC}^4)$ yield a common dense domain for all 
$\bmu_{\sigma,s}$ at varying of $s \in \mS$.
%
%
%
%
%
%
%
%
%
%

\begin{lem}
\label{lem.A.2}
The mapping $s \mapsto \bmu_{\sigma,s}$ is continuous in the Schwartz topology,
in the sense that for each $w,w' \in C^\infty_c(\bR^3,\hat{\bC}^4)$, the scalar products
\begin{equation}
\label{eq.00.14}
\int \lb w(\bsx) \, , \,\mp \sigma \star s_0(\bsx) w'(\bsx) \rb d^3 \bsx
\end{equation}
define tempered distributions in the argument $s_0 \in \mS(\bR^3)$.
\end{lem}
\begin{proof}
We set $\zeta_\mp(s_0)$ equals (\ref{eq.00.14}),
thus the Lemma will be proved if we show that the so-obtained linear functionals are 
tempered distributions.
Given $\bsx \in \bR^3$, we consider the compactly supported function
$s_0^\bsx \in \mS(\bR^3)$, $s_0^\bsx(\bsy) \doteq s_0(\bsy-\bsx)$.
It yields the explicit expression for the convolution
$\sigma \star s_0 (\bsx) \doteq \sigma(s_0^\bsx)$,
$\bsx \in \bR^3$,
that is a $C^\infty$ function \cite[Chap.9]{Con}.
Given a compact set $K \subset \bR^3$,
the modulus $| \sigma(s_0^\bsx) |$ has its maximum on some $\bsa \in K$.
In the rest of the proof we take as $K$ the closed support of $\lb w,w' \rb \in C_c^\infty(\bR^4,\bC)$.
Now, the continuity condition for $\sigma$ says that there are $c > 0$, $m \in \bN$
such that
\[
| \sigma(f) | \, \leq \, c \sum_{|\alpha| \leq m} \| (1+q_m) \partial^\alpha f \|_\infty 
\ \ \ , \ \ \ 
\forall f \in \mS(\bR^3)
\, ,
\]
where $q_m(\bsy) \doteq | \bsy |^m$ and $\alpha$ are multi-indices labelling the
partial derivatives $\partial^\alpha$.
We can now estimate
\begin{align}
| \zeta_\mp(s_0) | & = \ 
\left| \int \lb w(\bsx) , w'(\bsx) \rb \, \sigma \star s_0(\bsx) \, d^3\bsx \right| \\
& \leq \| \lb w,w' \rb \|_\infty 
    \int_K |\sigma\star s_0(\bsx)| \, d^3\bsx  \\ 
& = \| \lb w,w' \rb \|_\infty 
    \int_K | \, \sigma(s_0^\bsx) \, | \, d^3\bsx  \\ 
& \leq \| \lb w,w' \rb \|_\infty 
    {\mathrm{vol}}(K) \, | \, \sigma(s_0^\bsa) \, |  \\ 
& \leq c \, {\mathrm{vol}}(K) \| \lb w,w' \rb \|_\infty \cdot 
    \sum_{|\alpha| \leq m} \| (1+q_m) \partial^\alpha s_0^\bsa \|_\infty \\ 
& = c \, {\mathrm{vol}}(K) \| \lb w,w' \rb \|_\infty \cdot 
    \sum_{|\alpha| \leq m} 
    \left\| \frac{1+q_m}{1+q_m^\bsa} (1+q_m^\bsa) \partial^\alpha s_0^\bsa \right\|_\infty \\
& \leq c \, {\mathrm{vol}}(K) M \| \lb w,w' \rb \|_\infty \cdot 
    \sum_{|\alpha| \leq m} 
    \left\| (1+q_m^\bsa) \partial^\alpha s_0^\bsa \right\|_\infty  \\
& \leq c \, {\mathrm{vol}}(K) M \| \lb w,w' \rb \|_\infty \cdot 
    \sum_{|\alpha| \leq m} 
    \left\| (1+q_m) \partial^\alpha s_0 \right\|_\infty 
\, ,
\end{align}
where
%
%
$q_m^\bsa(\bsy) \doteq | \bsy - \bsa |^m$ and
$M \doteq \max_{ \bsa \in K } \| (1+q_m)(1+q_m^\bsa)^{-1} \|_\infty$.
%
%
%
Note that in the definition of $M$ we passed to the max over $\bsa$ because we want
an estimate independent of $s_0$ (note, in fact, that $\bsa$ depends on $s_0$);
that $M$ is well-defined follows by the fact that $\| (1+q_m)(1+q_m^\bsa)^{-1} \|_\infty$
is continuous on $\bsa \in K$.
In conclusion, the last estimate ensures that $\zeta_\mp$ fulfils the continuity property
with respect to the Schwartz topology, and the Lemma is proved.
\end{proof}

\medskip 

By construction we have the following equalities,

\begin{equation}
\label{eq.00.11}
[ \bmu_{\sigma,s} \, , \, \kappa ]_+ \ = \ 0
\ \ , \ \ 
[ u_{\sigma,s} \, , \,  \kappa ]  \ = \ 0
\ \ \ , \ \ \ 
\forall s \in \mS 
\, ,
\end{equation}

\medskip 

\noindent 
where $[\cdot , \cdot ]_+$ is the anticommutator
{\footnote{
Note that writing $u_{\sigma,s}$ as an exponential series we find
$u_{\sigma,s} \kappa = 
\sum_h (h!)^{-1} i^h \kappa \, (-1)^h \bmu_{s,\sigma}^h = 
\kappa u_{\sigma,s}$
by antilinearity of $\kappa$,
thus the first of (\ref{eq.00.11}) implies the second one.
}}.

\medskip 

We now consider the group $\sG \doteq \bR^3 \rtimes \bSO(3)$ 
of Galilean transformations generated by space translations and rotations.
It acts both on $\bsh$ and $\mS$
\[
U_\tau w(\bsx) \doteq w_+(\tau^{-1}\bs x) \oplus \bar{w}_-(\tau^{-1}\bs x)
\ \ \ , \ \ \ 
s_\tau(\bs x) \doteq s(\tau^{-1}\bs x)
\ \ \ , \ \ \ 
\tau \in \sG \, .
\]
This yields the unitary representation $U : \sG \to \mU(\bsh)$ and we find
\begin{equation}
\label{eq.cov}
U_\tau u_{\sigma,s} \ = \ u_{\sigma_R , s_\tau} U_\tau \, , 
\end{equation}
%
%
for $\tau = (\bs a , R) \in \sG$, $s \in \mS$.
Given the Weyl algebra $\mW$ defined by $(\mS,\eta)$,
we have an action by automorphisms $\alpha : \sG \to \Aut \mW$,
\[
\alpha_\tau(W(s)) \ \doteq \ W(s_\tau)
\ \ , \ \ 
\tau \in \sG 
\, .
\]
%
%
Now, we consider the free right Hilbert module $\efh \doteq \bsh \otimes \mW$, 
on which we define the right $\mW$-module operators
\begin{equation} 
\label{eq.00.00}
W^\lambda(s) (w \otimes A) \, \doteq \, u_{\sigma,s}w \otimes W(s)A \, ,
\end{equation}
$w \in \bsh$, $s \in \mS$, $A \in \mW$. 
The operators $W^\lambda(s)$ fulfil the Weyl relations, thus they induce a *-morphism
$\lambda : \mW \to \efB(\efh)$ defining a left $\mW$-module action on $\efh$.
With the method illustrated in \cite{Vas1}, we can define the fermionic Fock bimodule 
$\efF_a(\efh)$ and the C*-algebra $\mF_\sigma$ generated by operators 
$\psi(w)$, $W^\lambda(s)$, $w \in \bsh_+$, $s \in \mS$, acting on $\efF_a(\efh)$.
On $\mF_\sigma$ we have the relations

\begin{equation}
\label{eq.00.01}
W^\lambda(s) \psi(w) \ = \ \psi( e^{-i\sigma \star s_0} w ) W^\lambda(s) \, ,
\end{equation}

\

\noindent 
besides the usual CARs and Weyl relations for $\psi(w)$ and $W^\lambda(s)$ respectively
(here $\psi(w)$, $w \in \bsh_+$, is intended as the electron field \cite{Bon70}).
It is evident by (\ref{eq.00.01}) that $\sigma$ yields an obstacle to
make $\psi(w)$ and $W^\lambda(s)$ commute, thus we may regard $\mF_\sigma$
as a twist of the tensor product (that we obtain for $\sigma = 0$) of $\mW$ by 
the CAR algebra over $\bsh$.
In the following lines we provide a direct construction of Hilbert space representations
of $\mF_\sigma$, obtained in correspondence with states of $\mW$.

\paragraph{Representations and charged states.}
Let $\omega \in \mS(\mW)$ be a state with GNS triple 
$( \bsh^\omega , \Omega_b , \pi^\omega )$;
to be concise, we write
$W^\omega(s) \equiv \pi^\omega(W(s))$.
We set $\mH^\omega \doteq \mF_a(\bsh) \otimes \bsh^\omega$,
where $\mF_a(\bsh)$ is the fermionic Fock space
{\footnote{
For notions concerning the fermionic Fock space (2nd quantized operators,
creation and annihilation operators) we refer to \cite[Vol.2, \S 5.2.1]{BR}
(there, the subscript $-$ is used instead of $a$ for fermionic objects). 
}}.
On $\mH^\omega$ act the operators
$\psi(w) \otimes \bI^\omega$, $w \in \bsh_+$,
defined by the free (electron) Dirac field \cite{Bon70}, where $\bI^\omega \in \mB(\bsh^\omega)$ is the identity.
%
%
%
We set

\medskip 

\begin{equation}
\label{eq.00.02}
W^{\lambda,\omega}(s) \, \doteq \, \iGamma_a(u_{\sigma,s}) \otimes W^\omega(s)
\ \ \ , \ \ \ 
W^{\lambda,\omega}(s) \in \mU(\mH^\omega) \, ,
\end{equation}

\medskip 

\noindent 
where $\iGamma_a(u_{\sigma,s})$ is the fermionic 2nd quantization of 
$u_{\sigma,s} \in \mU(\bsh)$.
It is easily verified that the operators $W^{\lambda,\omega}(s)$ 
fulfil the Weyl relations, so (\ref{eq.00.02}) defines a representation 
$\pi^{\lambda,\omega}$ of $\mW$ that is reduced on each $n$-Fermion subspace
$\mH^{\omega,n} \doteq (\otimes_a^n \bsh) \otimes \bsh^\omega$, $n \in \bN$.
It is now easy to check that (\ref{eq.00.01}) holds in representation.
Given $z \in \bsh$, we consider the fermionic annihilation operators
$\bsa_f(z)$,
$\lb z |  \doteq (n)^{-1/2} \bsa_f(z)$,
in particular for $z = w \oplus 0$, $w \in \bsh_+$; then we take
$w_f \in \otimes^n_a \bsh \subset \mF_a(\bsh)$,
$v_b \in \bsh^\omega$
and compute
\[ 
\begin{array}{l}
W^{\lambda,\omega}(s) \psi(w) (w_f \otimes v_b)   \, = \\ 
(2)^{-1/2} W^{\lambda,\omega}(s) 
\{ \sqrt{n+1} (w \oplus 0) \otimes_a w_f + 
\sqrt{n} \lb 0 \oplus \bar{w}) | w_f \} \otimes W^\omega(s)v_b
\, = \\
%
(2)^{-1/2} \{ 
\sqrt{n+1} (e^{-i\sigma \star s_0}w \oplus 0) \otimes_a \iGamma_a(u_{\sigma,s})w_f + 
\sqrt{n} \lb 0 \oplus e^{i\sigma \star s_0}\bar{w}) | \iGamma_a(u_{\sigma,s})w_f \} 
\otimes W^\omega(s)v_b
\, = \\
%
%
%
\psi( e^{-i\sigma \star s_0} w ) W^{\lambda,\omega}(s) (w_f \otimes W^\omega(s)v_b) \, ,
\end{array}
\]
%
%
so that
\begin{equation}
\label{eq.00.01'}
W^{\lambda,\omega}(s) \psi(w) \ = \ \psi( e^{-i\sigma \star s_0} w ) W^{\lambda,\omega}(s) \, .
\end{equation}
We now introduce the notation $\vec{q} = ( q_i )$ to indicate a finite sequence with 
entries $\pm 1$; we write $|\vec{q}| \in \bN$ for the cardinality of $\vec{q}$ 
and $q \doteq \sum_i q_i \in \bZ$.
We use the above sequences to form Hilbert spaces $\bsh^{\vec{q}} \subset \mF_a^q(\bsh)$
spanned by vectors of the type 
\begin{equation}
\label{eq.00.12'}
\Omega^q_f \ \doteq \ 
\Omega^1_f \otimes_a \ldots \otimes_a \Omega^{|\vec{q}|}_f
\, ,
\end{equation}
where $\Omega^i_f \in \bsh_{ {\mathrm{sign}} (q_i) } \subset \bsh$
%
%
for all $i$.
By construction $\Omega^q_f \in \otimes_a^{|\vec{q}|} \bsh \cap \mF_a^q(\bsh)$, 
where $\mF_a^q(\bsh)$ is the subspace of $\mF_a(\bsh)$ with charge $q$.
Assuming that $\Omega^q_f$ is normalized, 
the vector $\Omega^q \doteq \Omega^q_f \otimes \Omega_b$ defines the state

\begin{equation}
\label{eq.00.12}
\omega_{\sigma,q}(W(s)) \ \doteq \
\lb \Omega^q \, , \, W^{\lambda,\omega}(s) \Omega^q \rb \ = \
\lb \Omega^q_f \, , \, \iGamma_a(u_{\sigma,s}) \Omega^q_f \rb \, \omega(W(s))
\, .
\end{equation}

\medskip 

\noindent States of the above type may coincide with $\omega$ on specific
localization regions, depending on the supports of $\sigma$ and $\Omega^q_f$.
For example, 
\begin{itemize}
\item If $\sigma$ is the Dirac delta then $\supp (\sigma \star s_0) = \supp(s_0)$,
      so that
      $\iGamma_a(u_{\sigma,s})\Omega^q_f = \Omega^q_f$
      and 
      $\omega_{\sigma,q}(W(s)) = \omega(W(s))$ for 
      $\supp(s) \cap \supp(\Omega^q_f) = \emptyset$.
      In other words, if we take $\omega$ as a reference state of $\mW$ then $\omega_{\sigma,q}$
      is localized on the support of $\Omega^q_f$.
\item If $\supp(\sigma)$ is contained in the 3--ball $B_r$ with radius $r$ centered in the origin 
      then 
      $\supp(\sigma \star s_0) \subseteq \supp(s_0) + B_r$ and we have
      $\omega_{\sigma,q}(W(s)) = \omega(W(s))$ whenever $s_0$ is such that
      $(\supp(\Omega^q_f) + B_r) \cap \supp(s_0) = \emptyset$.
      Thus the state $\omega_{\sigma,q}$ is localized in the region around 
      $\supp(\Omega^q_f)$ with radius $r>0$.
\item Analogously, if $\supp(\sigma) = \bR^3$ then there are no regions where
      $\omega_{\sigma,q} = \omega$.
      %
      %
\end{itemize}
The GNS representations of $\omega_{\sigma,q}$ are clearly given by the restrictions
of $\pi^{\lambda,\omega}$ to the cyclic subspaces of $\Omega^q$, which have charge $q \in \bZ$,
\begin{equation}
\label{eq.00.Hilb}
\begin{array}{l}
\mH^{\lambda,\omega,\Omega^q} \ \doteq \ 
{\mathrm{span}} \, 
\{ W^{\lambda,\omega}(s) \Omega^q = \iGamma_a(u_{\sigma,s})\Omega^q_f \otimes W(s)\Omega_b \, , \, 
s \in \mS \} 
\, , \\ \\
\mH^{\lambda,\omega,\Omega^q} \, \subset \, 
\mH^{\omega,q} \doteq \mF_a^q(\bsh) \otimes \bsh^\omega \, .
\end{array}
\end{equation}

\noindent 
Note that the spaces $\mH^{\lambda,\omega,\Omega^q}$ do not have the form of 
a tensor product, because the test functions $s \in \mS$ stand simultaneously 
in both the factors. 
%
%
We write, at varying of $q \in \bZ$,
\begin{equation}
\label{eq.00.Hilb'}
\pi^{\lambda,\omega,q} \doteq \pi^{\lambda,\omega} \restriction \mH^{\omega,q} 
\ \ \ , \ \ \
\pi^{\lambda,\omega,\Omega^q} \doteq \pi^{\lambda,\omega} \restriction \mH^{\lambda,\omega,\Omega^q}
\, .
\end{equation}
Since $\psi(w)$ maps $\mF_a^q(\bsh)$ into $\mF_a^{q-1}(\bsh)$,
by (\ref{eq.00.01'}) we find
\begin{equation}
\label{eq.00.pq}
\psi(e^{-i\sigma \star s_0} w) W^{\lambda,\omega,q}(s) \ = \ 
W^{\lambda,\omega,q-1}(s) \psi(w) \, ,
\end{equation}
where $W^{\lambda,\omega,q}(s) \doteq \pi^{\lambda,\omega,q}(W(s))$.

\medskip 

We wish to understand the interplay between the $\sG$-action and $\pi^{\lambda,\omega}$.
To this end, we assume that the reference state $\omega$ is $\sG$-invariant
(for example, the Fock state), so that there is a unitary representation
$U^\omega : \sG \to \mU(\bsh^\omega)$ besides 
$U^a : \sG \to \mU(\mF_a(\bsh))$ which is defined by fermionic second quantization
$U^a_\tau \doteq \iGamma_a(U_\tau)$.
Applying (\ref{eq.cov}) we get
\[
\begin{array}{lcl}
\pi^{\lambda,\omega}(W(s_\tau)) & = &
\iGamma_a(u_{\sigma,s_\tau}) \otimes W^\omega(s_\tau) \ =  \\ & = & 
U^a_\tau \iGamma_a(u_{\sigma_{R^{-1}},s}) U^{-*}_\tau \otimes 
U^\omega_\tau W^\omega(s) U^{\omega *}_\tau \ =  \\ & = & 
V_\tau \, \pi^{\omega,\sigma_{R^{-1}}}(W(s)) \, V_\tau^*
\, ,
\end{array}
\]
having defined $V_\tau \doteq U^a_\tau \otimes U^\omega_\tau$ and 
$\pi^{\omega,\sigma_{R^{-1}}}(W(s)) \doteq \iGamma_a(u_{\sigma_{R^{-1}},s}) \otimes W^\omega(s)$.
The previous expression says that:
(1) the representations $\pi^{\lambda,\omega}$ and $\pi^{\omega,\sigma_{R^{-1}}}$ 
are unitarily equivalent for any $\tau = (\bsa,R)$, and if $\sigma$ is rotation-invariant
then $\pi^{\lambda,\omega}$ is $\sG$-covariant;
(2) for generic $\sigma$, the representation $\pi^{\lambda,\omega}$ is translation covariant,
\begin{equation}
\label{eq.00.15}
\pi^{\lambda,\omega}(W(s_\bsa)) \ = \ 
V_\bsa \, \pi^{\lambda,\omega}(W(s)) \, V_\bsa^*
\ .
\end{equation}
The same computation yields
\begin{equation}
\label{eq.00.15'}
\omega_{\sigma,q}(W(s_\tau)) 
\ = \ 
\lb U^{-*}_\tau\Omega^q_f \, , \, \iGamma_a(u_{\sigma_{R^{-1}},s}) U^{-*}_\tau\Omega^q_f \rb \, 
\omega(W(s))
\ .
\end{equation}

\paragraph{Regular representations.}
Assuming that $\pi^\omega$ is regular,
by (\ref{eq.lem.A.1}) and (\ref{eq.00.02}) we conclude that $\pi^{\lambda,\omega}$
is regular, thus we introduce the (Segal) field
$\varphi^\lambda$ defined by $W^{\lambda,\omega}(s) = e^{i\varphi^\lambda(s)}$. 
Explicitly, we find
\begin{equation}
\label{eq.00.03}
\varphi^\lambda(s) \, = \, 
\bI_f \otimes \varphi(s) + \bmu_\sigma^a(s) \otimes \bI_b  \, ,
\end{equation}
where $\varphi$ is the field defined by the regular representation $\pi^\omega$ and
\begin{equation}
\label{eq.00.bmu}
\bmu_\sigma^a : \mS \to \mL(\mD) \ \ \ , \ \ \ \mD \subset \mF_a(\bsh) \, ,
\end{equation}
is defined by $\bmu_\sigma^a(s) \doteq \iGamma_a(\bmu_{\sigma,s})$
with $\bmu_{\sigma,s}$ the Stone generator (\ref{eq.00.08});
in the above expression, $\mD$ is the subspace of antisymmetric tensor products 
of spinors $w \in C_c^\infty(\bR^3,\hat{\bC}^4)$ and $\mL(\mD)$ 
the *-algebra of linear operators from $\mD$ into $\mD$.
By Lemma \ref{lem.A.2}, $\bmu_\sigma^a$ is an operator-valued 
distribution. 
We shall refer to $\varphi^\lambda$ as the \emph{twisted bosonic field}.
By (\ref{eq.00.11}), and using the fact that $\bmu_\sigma^a$ by construction
preserves the fermionic $n$-particle space (and in particular both the 
positive and negative charge subspaces), we find
\begin{equation}
\label{eq.00.11'}
[ \bmu_\sigma^a \, , \, N_f ] \ = \ 
[ \bmu_\sigma^a \, , \, Q ] \ = \ 
[ \bmu_\sigma^a \, , \, \iGamma_a(\kappa) ]_+ \ = \ 0 \, ,
\end{equation}
where $N_f$ is the Fermion particle number operator,
$Q$ is the electron-positron charge and $\iGamma_a(\kappa)$
is the fermionic 2nd quantized of $\kappa$.

\medskip

To give an explicit expression of $\bmu_\sigma^a$, for any $n \in \bN$ and $k=1,\ldots,n$ 
we introduce the matrix
\[
\tau_{n,k}
\ \doteq \ 
\ldots {\bf 1}_8 \otimes 
\stackrel{k}{
\widehat{
\left(
\begin{array}{cc}
- {\bf 1}_4 & 0 \\
0 & {\bf 1}_4
\end{array}
\right)
}}
\otimes {\bf 1}_8 \ldots \, ,
\]
acting non-trivially only on the $k$th factor of $\otimes^n \hat{\bC}^4$;
it encodes the action of the electron-positron charge on the given factor.
%
%
Then, given 
the antisymmetric projection $P_n^a$ onto $\otimes_a^n \hat{\bC}^4$ and
$w_f = \oplus_n w_f^n \in \mD$, $w_f^n \doteq P_n^a w^n$, we write

\begin{equation}
\label{eq.00.10}
\bmu_\sigma^a(s) w_f(\vec\bsy) \ = \ 
\bigoplus_n P_n^a \sum_{k=1}^n \int 
\tau_{n,k} \, \sigma(\bsy_k - \bsx) s_0(\bsx) w_f^n(\vec\bsy) 
\, d^3\bsx
\, ,
\end{equation}

\medskip 

\noindent where
$\vec\bsy \doteq (\bsy_1,\ldots,\bsy_k,\ldots,\bsy_n) \in \bR^{3n}$,
$w_f^n \in C_c^\infty(\bR^{3n},\otimes^n \hat{\bC}^4)$ is 
antisymmetric in the arguments $\bsy_1,\ldots$, $\bsy_n \in \bR^3$.
Thus, regarded as a "field", $\bmu_\sigma^a$ is given by
%
%
\begin{equation}
\label{eq.00.10'}
\bmu_\sigma^a(\bsx)  \ = \ 
\bigoplus_n P_n^a \sum_{k=1}^n
\tau_{n,k} \, 
\sigma_{n,k,\bsx}
\, ,
\end{equation}


\noindent having defined the operator-valued distribution 
$\sigma_{n,k,\bsx} w_f^n(\vec\bsy) \doteq \sigma( \bsx + \bsy_k ) 
w_f^n(\vec\bsy)$
on the dense subspace of $\otimes^n_a \bsh$ spanned by antisymmetric functions
in $C_c^\infty(\bR^{3n},\otimes^n \hat{\bC}^4)$.

\medskip 

\begin{ex}[The Fock state]
Let $\omega(W(s)) \doteq e^{-1/4 (\| s_0 + is_1 \|_2^2)}$, $s \in \mS$,
denote the Fock state.
Then $\bsh^\omega$ is the bosonic Fock space 
$\mF_s(\bsh_0)$, $\bsh_0 \doteq L^2(\bR^3,\bC)$,
and $\varphi$ is the bosonic free field over $\bsh_0$.
If $\phi$ is the usual initial condition of the free scalar field and $\pi$
its conjugate field \cite[\S 8.4.A]{BLOT}, then
$\phi(s_0) = \varphi(s_0,0)$, $\pi(s_1) = \varphi(0,s_1)$.
%
%
%
%
%
%
%
%
This yields the Fock representation of $\mW$ over $\mF_s(\bsh_0)$.
%
%
%
%
Adding $\bmu_\sigma^a$ for an arbitrary distribution $\sigma$
yields the fields $\varphi^\lambda$.
The Weyl operators in representation are given by
\[
W^{\lambda,\omega}(s)  \, \doteq \, e^{ i\varphi^\lambda(s) } \, = \, e^{i\bmu_\sigma^a(s)} \otimes e^{i\varphi(s)} \, .
\]
Given the states 
$\omega_{\sigma,q} \in \mS(\mW)$ (\ref{eq.00.12}), $q \in \bZ$,
we find 
%
%
%
%
\[
\omega_{\sigma,q}(W(s)) \ = \  
\frac{1}{|\vec{q}|!} \, \sum_\pi \prod_k 
\int\lb \Omega^i_f(\bsx) , e^{-i \sigma\star s_0(\bsx)} \Omega^{\pi(i)}_f(\bsx) \rb \, d^3\bsx 
\, \cdot \,  
e^{-1/4 \| s_0 + i s_1 \|_2^2 }
\, .
\]
The localization properties of the above states depend on $\sigma$
as discussed in the previous paragraph. For example, in the case 
$q=-1$, $\sigma(\bsx) = 1/|\bsx|$, so that $\Omega^1_f \in \bsh_+$,
\begin{equation}
\label{eq.00.13}
\omega_{\sigma,1}(W(s)) \ = \  
\int \lb \Omega^1_f(\bsx) , e^{-i \int |\bsx - \bsy|^{-1} s_0(\bsy) d^3\bsy } \Omega^1_f(\bsx) \rb 
\, d^3\bsx \cdot 
e^{-1/4 \| s_0 + i s_1 \|_2^2 }
\ .
\end{equation}
%
%
\end{ex}

\paragraph{Elementary properties of the twisted bosonic field.}
Our aim is now to give some elementary properties of the twisted field $\varphi^\lambda$
defined in (\ref{eq.00.03}). 
As a first step we note that 
%
for all $s,s' \in \mS$ we have
\[
[ \varphi^\lambda(s) \, , \, \varphi^\lambda(s') ] \ = \ 
[ \varphi(s) \, , \, \varphi(s') ] \ = \ 
-i\eta(s,s') \, ,
\]
where $\varphi$ is the untwisted bosonic field defined by the regular reference state 
$\omega \in \mS(\mW)$.
Assuming that $\varphi$ is constructed in terms of creation-annihilation operators 
$\bsa_b^\circ(s)$,
we may write 
\begin{equation}
\label{eq.tbf01}
\varphi^\lambda(s) \ = \ 
\frac{1}{\sqrt{2}} ( \bsa_{b,\sigma}(s) + \bsa_{b,\sigma}^*(s) )
\end{equation}
where
\begin{equation}
\label{eq.tbf02}
\bsa_{b,\sigma}^\circ(s) \, \doteq \, 
\bI_f \otimes \bsa_b^\circ(s) + \frac{1}{\sqrt{2}} \bmu_\sigma^a(s) \otimes \bI_b 
\, ,
\end{equation}

\noindent 
and $\circ$ indicates a blank or the * symbol.
Starting from the hypothesis that $\bsa_b(s) \Omega_b = 0$, we find
\[
\bsa_{b,\sigma}(s) \, \Omega \ = \ 0 \, ,
\]
where $\Omega \doteq \Omega_f \otimes \Omega_b$ and $\Omega_f$ is the 
fermionic Fock vacuum. Thus
$\lb \Omega , \varphi^\lambda(s) \Omega \rb = 0$
provided that the same is true for $\varphi$ with respect to $\Omega_b$
(that is the case when $\omega$ is the Fock state).
%

\medskip 

In accord with the considerations of the previous paragraphs, 
$\varphi^\lambda$ is highly reducible on $\mH^\omega$.
The fact that $\varphi^\lambda$ commutes with the fermionic number operator
implies that we have restricted fields $\varphi^{\lambda,n}$ defined on 
the Hilbert spaces $\mH^{\omega,n} \doteq (\otimes_a^n \bsh) \otimes \bsh^\omega$, 
$n \in \bN$, so that we may write
$\varphi^\lambda =  \oplus_n \varphi^{\lambda,n}$,
$\mH^\omega  = \oplus_n \mH^{\omega,n}$.
By definition $\varphi^{\lambda,0}$ is the untwisted field $\varphi$
acting irreducibly on $\mH^{\omega,0} \simeq \bsh^\omega$.
Instead, the fields $\varphi^{\lambda,n}$ can in turn be reduced on
the cyclic subspaces 
$\mH^{\lambda,\omega,\Omega^q} \subset \mH^{\omega,n}$, $n = |\vec{q}|$ (\ref{eq.00.Hilb}).
An alternative decomposition is $\varphi^\lambda = \oplus_{q \in \bZ} \varphi^{\lambda,q}$
according to the electron-positron charge, obtained by reducing $\varphi^\lambda$
over the $q$ charge subspaces $\mH^{\omega,q}$ (\ref{eq.00.Hilb});
they correspond to the representations $\pi^{\lambda,\omega,q}$ of the 
Weyl algebra (\ref{eq.00.Hilb'}).

\paragraph{Expectation values.}
Let $n = |\vec{q}| \neq 0$ and $\Omega^q = \Omega^q_f \otimes \Omega_b \in \mH^{\omega,n}$
a normalized vector of the type (\ref{eq.00.12'}), where $\Omega^q_f \in \otimes_a^n \bsh$ 
has compact support and charge $q$. 
%
%
We consider, at varying of $m \in \bN$, the Hilbert spaces
\begin{equation}
\label{eq.tbf03}
\mH^\lambda_m \ \doteq \ 
{\mathrm{span}}\{ \varphi^\lambda(s_1) \cdots \varphi^\lambda(s_m) \Omega^q \}
\, \subseteq \mH^{\omega,n} \, ,
\end{equation}
%
%
%
%
%
and the expectation values
\begin{equation}
\label{eq.weyl.vev1}
{\bs w}^{\lambda,q}_m(s^1 , \ldots , s^m) 
\ \doteq \ 
\lb \Omega^q \, , \, \varphi^\lambda(s^1) \cdots \varphi^\lambda(s^m) \Omega^q \rb 
\ \ \ , \ \ \ 
s^1 , \ldots , s^m \in \mS \, .
\end{equation}
We analyze the case $n=1$, $q=-1$, so that we have a vector
$\Omega^1_f \in C_c^\infty(\bR^3,\bC^4) \subset \bsh_+$ and 
$
\bmu_\sigma^a(s) \Omega^1_f = 
\bmu_{\sigma,s} \Omega^1_f =
- \sigma \star s_0 \cdot \Omega^1_f$.
%
Writing 
\[
{\bs w}_m(s^1 , \ldots , s^m)
\ \doteq \ 
\lb \Omega^1 \, , \, \varphi(s^1) \cdots \varphi(s^m) \Omega^1 \rb
\ = \ 
\lb \Omega_b \, , \, \varphi(s^1) \cdots \varphi(s^m) \Omega_b \rb
\]
and 
\[
{\bs w}^\sigma_m(s^1,\ldots,s^m) \ \doteq \  
\lb \Omega^1 , \bmu_\sigma^a(s^1) \cdots \bmu_\sigma^a(s^m) \Omega^1 \rb \ = \ 
\lb \Omega^1_f , \bmu_{\sigma,s^1} \cdots \bmu_{\sigma,s^m} \Omega^1_f \rb \, ,
\]
we have the "$m$-point functions"
%
%
\[
{\bs w}^{\lambda,1}_m(s^1 , \ldots , s^m) 
\ = \ 
\sum_k \sum_{ \bsm = \bsi \cup \bsj }
{\bs w}_{m-k}(s^{i(1)},\ldots\,s^{i(m-k)}) \, {\bs w}^\sigma_k(s^{j(1)},\ldots\,s^{j(k)})
\, ,
\]
where $\bsm = \{ 1 , \ldots m \}$ and $\bsm = \bsi \cup \bsj$ is any disjoint union
with $\bsi = \{ i(1) < \ldots < i(m-k) \}$ and $\bsj = \{ j(1) < \ldots < j(k) \}$.
Note that we may have ${\bs w}_1 \neq 0$, for example when
$\omega(W(s)) = e^{i\rho(f) - 1/4 \| s_0 + i s_1 \|_2^2}$ for some
distribution $\rho \in \mS'(\bR^3)$ (in this case ${\bs w}_1(s) = \rho(s)$).
Also note that in general we have
\[
{\bs w}^\sigma_1(s) \, = \, 
\lb \Omega^1_f ,  \bmu_{\sigma,s} \Omega^1_f \rb \, \neq \, 0 
\]
for $\Omega^1_f$ having support overlapping with that of $\sigma \star s$.
When $\omega$ is the Fock state, the odd terms ${\bs w}_{2k+1}$ vanish so that 
%
%
%
%
\[
{\bs w}^{\lambda,1}_{2m+1}(s^1 , \ldots , s^{2m+1})
\, = \, 
\sum_{k \ {\mathrm{odd}}} \sum_{ 2\bsm+1 = \bsi \cup \bsj }
{\bs w}_{2m+1-k}(s^{i(1)},\ldots\,s^{i(2m+1-k)}) \, 
{\bs w}^\sigma_k(s^{j(1)},\ldots\,s^{j(k)})
\, ,
\]
and in particular ${\bs w}^{\lambda,1}_1 = {\bs w}^\sigma_1 \neq 0$.
Thus there are non-trivial transition amplitudes between states belonging to the spaces
$\mH^\lambda_m$ and $\mH^\lambda_{m+1}$, in contrast with the case where $\sigma = 0$ 
and $\omega$ the Fock state.
Of course, analogous considerations hold for generic  $q \in \bZ$.
Since the hierarchy of distributions ${\bs w}^{\lambda,q}$ 
is invariant for unitary equivalence of representations of $\mW$, we proved:
\begin{prop}
Let $\omega$ be the Fock space on $\mW$ and $\sigma \neq 0$. 
Then, given a compactly supported fermionic vector $\Omega^q_f$ with charge $q \in \bZ$, 
the cyclic representation $\pi^{\lambda,\omega,\Omega^q}$ 
in (\ref{eq.00.Hilb'}) is not unitarily equivalent to $\pi^\omega$. 
\end{prop}

Let $R \subset \bR^3$ such that 
$R \cap (\supp(\Omega^q_f) + \supp(\sigma)) = \emptyset$.
Then ${\bs w}^\sigma_m(s^1,\ldots,s^m) = 0$ for all $s^1,\ldots,s^m$
supported in $R$ and this indicates that $\pi^{\lambda,\omega,\Omega^q}$
is unitarily equivalent to $\pi^\omega$ on such a region.
Yet, as we shall see in the following section,
in cases of physical interest we have $\supp(\sigma) = \bR^3$.

\medskip 

Now, it is well-known that distributions $\rho \in \mS'(\bR^3)$ induce
Weyl algebra automorphisms $\beta_\rho(W(s)) \doteq e^{i\rho(s)}W(s)$, $W(s) \in \mW$, 
such that in the regular representation $\pi^\omega \beta_\rho$ the corresponding 
Segal field is $\varphi_\rho \doteq \varphi + \rho \bI_b$. 
The interpretation of $\rho$ is then that it is an external field interacting with $\varphi$.
Picking $\rho = {\bs w}^\sigma_1$ we have that $\varphi^\lambda$ and $\varphi_\rho$ 
have the same "one-point function", yet for $m = 2$ we find
\[
\lb \Omega_b , \varphi_\rho(s^1) \varphi_\rho(s^2) \Omega_b \rb 
\ = \ 
{\bs w}_2(s^1,s^2) + {\bs w}^\sigma_1(s^1) {\bs w}^\sigma_1(s^2) \, .
\]
The above distribution in general does not equal ${\bs w}^{\lambda,1}_2$ 
because of the term 
${\bs w}^\sigma_2(s^1,s^2)$,
which differs from ${\bs w}^\sigma_1(s^1) {\bs w}^\sigma_1(s^2)$ unless
the operators $\bmu_{\sigma,s^1}$, $\bmu_{\sigma,s^2}$ are given by 
multiplication by a scalar.
This proves that in general our representations cannot be interpreted in 
terms of external potentials.

\paragraph{Infinitesimal commutation relations.}
In the sequel, to be concise we will write
\[
\bmu_\sigma^a(s) \equiv \bmu_\sigma^a(s) \otimes \bI_b
\ \ , \ \
\varphi(s) \equiv \bI_f \otimes \varphi(s) 
\ \ , \ \ 
\psi(w) \equiv \psi(w) \otimes \bI_b
\, ,
\]
so that 
$\varphi^\lambda = \varphi + \bmu_\sigma^a$.
Our aim is to check the infinitesimal version of (\ref{eq.00.01'}),
\begin{equation}
\label{eq.00.06}
[ \varphi^\lambda(s) \, , \,  \psi(w) ]
\ =  \ 
-\psi(\sigma \star s_0 \cdot w )
\, .
\end{equation}

\medskip

\noindent 
To this end, we proceed by evaluating the commutator at l.h.s. of (\ref{eq.00.06})
on vectors of the type $w_f \otimes v_b$, where
$v_b \in \Dom(\varphi) \subset \bsh^\omega$ and
$w_f \doteq w_1 \otimes_a \ldots \otimes_a w_N \in 
 \otimes_a^N \bsh \cap \Dom (\bmu_\sigma^a(s))$
can be written for convenience
\[
w_f = \frac{1}{N} \sum_k (-1)^k w_k \otimes w^a_k
\ \ \ , \ \ \ 
w^a_k \in \otimes_a^{N-1} \bsh
\, .
\]
Since 
$[ \varphi(s) \, , \,  \psi(w) ] = 0$, 
we find
$[ \varphi^\lambda(s) \, , \,  \psi(w) ] 
=
[ \bmu_\sigma^a(s) , \psi(w) ]$ 
and, given the fermionic creation and annihilation operators 
$\bsa_f^*(w)$, $\bsa_f(w)$, $w \in \bsh$, we perform the standard computations 
%

\[
\begin{array}{lccl}
\bmu_\sigma^a(s) \bsa_f(\kappa w) w_f
& = & &
%
n^{-1/2} \sum_k (-1)^{k-1}  
\left( 
\lb \kappa w , \bmu_{\sigma,s} w_k \rb \, w^a_k
+
\lb \kappa w , w_k \rb \, \bmu_\sigma^a(s) w^a_k
\right) + 
\,  \\ & & - &
n^{-1/2} \sum_k (-1)^{k-1}  
\lb \bmu_{\sigma,s} \kappa w , w_k \rb \, w^a_k
%
%
%
\, = \\ & = &  & 
\bsa_f(\kappa w) \bmu_\sigma^a(s) w_f + \bsa_f(\kappa \bmu_{\sigma,s} w) w_f \, ,
\end{array}
\]

\[
\begin{array}{lcl}
\bmu_\sigma^a(s) \bsa_f^*(w) w_f
& = & 
\sqrt{n+1} \bmu_{\sigma,s}w \otimes_a w_f +
\sqrt{n+1} w \otimes_a \sum_k \ldots \otimes_a  \bmu_{\sigma,s}w_k \otimes_a \ldots 
\, = \\ & = & 
\bsa_f^*(\bmu_{\sigma,s}w) w_f + \bsa_f^*(w) \bmu_\sigma^a(s) w_f \, ,
\end{array}
\]
%
%
having used (\ref{eq.00.11}). Summing up the two equalities we get
\begin{equation}
\label{eq.00.06'}
[ \bmu_\sigma^a(s) \, , \,  \hat{\psi}(w) ]
\ = \ 
\hat{\psi}(\bmu_{\sigma,s}w)
\, ,
\end{equation}
where $\hat{\psi}$ stands for the self-dual Dirac field.
For $w \in \bsh_+ \oplus 0$, we have $\hat{\psi}(w) = \psi(w)$ and $\bmu_{\sigma,s}w = -(\sigma \star s_0)w$;
thus we get (\ref{eq.00.06}), corresponding to the expression
%
%

\begin{equation}
\label{eq.00.07}
[ \varphi^\lambda(\bsx) \, , \,  \psi(\bsy) ]
\ =  \ 
- \sigma(\bsy - \bsx) \psi(\bsy)
\, .
\end{equation}

\paragraph{Differential operators.}
Let 
$P_\partial \, $ 
denote a differential operator with constant coefficients
(e.g. $P_\partial = -\bs\Delta$, the Laplace operator).
Assuming that $\sigma$ is a fundamental solution of $P_\partial \, $
%
%
we have $P_\partial \sigma = \delta$, that is,
$P_\partial \,  (\sigma \star s_0) = s_0$, and by elementary properties of convolutions we find
$
P_\partial \,  \bmu_{\sigma,s} \doteq 
\bmu_{\sigma , P_\partial \,  s} =
\bmu_{P_\partial \, \sigma , s} =
\bmu_{\delta,s}
$.
%
%
Since $\delta \star s_0 = s_0$, for $w = w_+ \oplus \bar{w}_- \in \bsh$ we have
\[
P_\partial \,  \mu_{\sigma,s} ( w_+ \oplus \bar{w}_- ) \ = \ 
\bmu_{\delta,s}w \ = \ 
(-s_0 w_+) \oplus (s_0 \bar{w}_-) \ = \ 
\tau s_0 w
\ \ \ , \ \ \ 
\tau
\ \doteq \ 
\left(
\begin{array}{cc}
- {\bf 1}_4 & 0 \\
0 & {\bf 1}_4
\end{array}
\right)
\, .
\]
%
%
As we saw in the previous paragraphs, $\bmu_{\bsdelta,s}$ is the 
infinitesimal counterpart of the unitary representation 
\[
u_{\bsdelta,s}w(\bsx) \ \doteq \ e^{-is_0}w_+ \oplus e^{is_0}\bar{w}_-
\]
inducing local gauge transformations of the $\psi$ field.
Given $\bmu_\bsdelta^a(s) \doteq \iGamma_a(\bmu_{\bsdelta,s})$,
to be concise we write
\[
\varrho(s) \ \doteq \ 
\varphi^\lambda(P_\partial s) \ = \ 
\varphi(P_\partial s) + \bmu_\bsdelta^a(s) \, .
\]
By applying (\ref{eq.00.06}), and using the fact that 
$[ \varphi(s) \, , \,  \psi(w) ] = 0$, 
we obtain the following relations on vectors in domain and 
for $w \in C^\infty_c(\bR^3,\bC^4) \subset \bsh_+$:

\begin{equation}
\label{eq.00.18}
[ \varrho(s) \, , \, \psi(w) ]
\ = \ 
- \psi(s_0 w)
\ \ \ , \ \ \ 
e^{i\varrho(s)} \psi(w) e^{-i\varrho(s)}
\ = \ 
\psi(e^{-i s_0} w) \, .
\end{equation}

\medskip 

\noindent Thus $\psi(w)$ and $\varrho(s)$ commute for $\supp(s) \cap \supp(w) = \emptyset$,
independently of the commutation relations between $\psi$ and $\varphi^\lambda$.
As a consequence, if we evaluate the states (\ref{eq.00.12}) over the C*-subalgebra $\mW_\sigma$ 
generated by Weyl unitaries $W(P_\partial s)$, then we find
\[
\omega_{\sigma,q}(W(P_\partial s)) \ = \ 
\lb \Omega^q_f \, , \, \iGamma_a(u_{\bsdelta,s}) \Omega^q_f \rb \, \omega(W(P_\delta s))
\, ,
\]
implying 
\[
\omega_{\sigma,q}(W(P_\partial s)) \ = \ \omega(W(P_\delta s))
\ \ \ , \ \ \
\supp(s) \cap \supp(\Omega^q_f) = \emptyset \, .
\]
We conclude that the restriction of $\omega_{\sigma,q}$ to $\mW_\sigma$ is localized on
the support of $\Omega^q_f$.


\section{Synthesis and examples}
\label{sec.syn}

%
%
\begin{thm}
\label{thm.00.1}
Let $\mW$ denote the Weyl algebra of be the symplectic space $(\mS,\eta)$ (\ref{eq.00.i1}), and
$\omega \in \mS(\mW)$ a regular state with GNS triple $( \bsh^\omega , \Omega_b , \pi^\omega )$.
Let $\bsh = \bsh_+ \oplus \bsh_-$ be the Hilbert space with conjugation $\kappa$ (\ref{eq.00.i2}).
Then for any distribution $\sigma \in \mS'(\bR^3)$ the following properties hold.
\begin{enumerate}
    \item There are a Dirac (electron) field $\psi$ and 
          a bosonic Field $\varphi^\lambda$,
          both defined on the Hilbert space 
          $\mH^\omega \doteq \mF_a(\bsh) \otimes \bsh^\omega$
          and fulfilling, besides the usual C(A/C)Rs, the commutation relations
          \begin{equation}
          \label{eq.thm.00.1.1}
          [ \varphi^\lambda(s) , \psi(w) ] \ = \ 
          -\psi(\sigma \star s_0 \cdot w) 
          \end{equation}
          for all $w \in C^\infty_c(\bR^3,\bC^4) \subset \bsh_+$,
          $s=(s_0,s_1),s' \in \mS$. In terms of fields,
          \begin{equation}
          \label{eq.thm.00.1.1'}
          [ \varphi^\lambda(\bsx) , \psi(\bsy) ] \ = \ 
          -\sigma(\bsy - \bsx) \psi(\bsy) \, .
          \end{equation}
    \item If $\sigma$ is a fundamental solution of the differential operator
          $P_\partial$, then defining $\varrho \doteq \varphi^\lambda P_\partial$
          we have
          \begin{equation}
          \label{eq.thm.00.1.2}
          [ \varrho(s) , \psi(w) ] \, = \, 
          -\psi(s_0 w) 
          \ \ \ , \ \ \ 
          e^{i\varrho(s)} \psi(w) e^{-i\varrho(s)} \, = \, 
          \psi(e^{-is_0}w)
          \, .
          \end{equation}
          That is, $\varrho$ is a generator of
          the local gauge transformations $\psi(w) \to \psi(e^{-is_0}w)$.
          The operators $\varrho(s)$ and $\psi(w)$ commute for 
          $\supp(s) \cap \supp(w) = \emptyset$.
    \item Let $q \in \bZ$ and $\Omega^q_f \in \mF_a^q(\bsh)$ denote a product vector 
          (\ref{eq.00.12'}).
          Then there is a representation
          $\pi^{\lambda,\omega,\Omega^q} : \mW \to \mB(\mH^{\lambda,\omega,\Omega^q})$,
          $\mH^{\lambda,\omega,\Omega^q} \subset \mH^\omega$,
          having cyclic vector $\Omega^q \doteq \Omega^q_f \otimes \Omega_b$.
          It turns out that $\pi^{\lambda,\omega,\Omega^q}$ is the GNS representation
          of the state
          \begin{equation}
          \label{eq.thm.00.1.3}
          \omega_{\sigma,q}(W(s)) \ \doteq \ 
          \lb \Omega^q_f , \iGamma_a(u_{\sigma,s})\Omega^q_f \rb \, 
          \omega(W(s))
          \, ,
          \end{equation}
          where $\iGamma_a(u_{\sigma,s})$ is the fermionic 2nd quantized of 
          $u_{\sigma,s}w \doteq 
          e^{-i\sigma \star s_0}w_+ \oplus e^{i\sigma \star s_0} \bar{w}_-$, 
          $w = w_+ \oplus \bar{w}_- \in \bsh$.
          The states $\omega_{\sigma,q}$ are localized in the region 
          $\supp(\Omega^q_f) + \supp(\sigma)$ in the sense that 
          \begin{equation}
          \label{eq.thm.00.1.4}
          \omega_{\sigma,q}(W(s)) \, = \, \omega(W(s)) 
          \ \ \ , \ \ \ 
          \supp(s) \cap ( \supp(\Omega^q_f) + \supp(\sigma) ) = \emptyset \, ,
          \end{equation}
          yet their restrictions on $\mW_\sigma$ are localized in $\supp(\Omega^q_f)$,
          \begin{equation}
          \label{eq.thm.00.1.4'}
          \omega_{\sigma,q}(W(P_\partial s)) \, = \, \omega(W(s)) 
          \ \ \ , \ \ \ 
          \supp(s) \cap \supp(\Omega^q_f) = \emptyset \, .
          \end{equation}
          %
\end{enumerate}
\end{thm}


\ 

\noindent \emph{Example 1. Toy models with the Lebesgue measure.}
We consider the distribution $\sigma(\bsx) \equiv 1$ 
given by the Lebesgue measure. 
Applying the convolution by a test function $s_0 \in \mS(\bR^3)$ yields
$\sigma\star s_0(\bsx) = -\lb s_0 \rb$ where $\lb s_0 \rb \in \bR$ is the 
integral of $s_0$.
Given a regular GNS triple $( \bsh^\omega , \Omega_b , \pi^\omega )$ for 
$\omega \in \mS(\mW)$, where $\mW$ is the Weyl algebra (\ref{eq.00.i1}), 
we form the Hilbert space $\mH^\omega \doteq \mF_a(\bsh) \otimes \bsh^\omega$, 
on which the fields $\varphi^\lambda$ and $\psi$ are defined.
They fulfil the commutation relations
\[
[ \varphi^\lambda(s) \, , \, \psi(w)  ] \ = \ \lb s_0 \rb \psi(w)
\, ,
\]
thus test functions $s \in \mS$ with $\lb s_0 \rb = \pm 1$ define field operators that
detect the electron-positron charge.
The explicit expression of $\varphi^\lambda$ is given by 
$\varphi^\lambda(s) = \varphi(s) + \bmu_\sigma^a(s)$
where $\varphi$ is the bosonic field defined by $\omega$ and
$\bmu_\sigma^a(s)$ the fermionic second quantization of the 
multiplication operator by $-\lb s_0 \rb$.
We have the decomposition $\varphi^\lambda = \oplus_{n \in \bN} \varphi^{\lambda,n}$
with respect to the fermion number operator,
\[
\varphi^{\lambda,0} = \varphi 
\ \ \ , \ \ \ 
\varphi^{\lambda,n}(s) = 
\varphi(s) - P_n^a \sum_{k=1}^n \tau_{n,k} \lb s_0 \rb 
\ \ , \ \ 
n \in \bN 
\, .
\]
Taking vectors of the type 
$w^q_f = w_1 \otimes_a \ldots \otimes_a w_n$
with charge $q \in \bZ$ and such that $w_k \in \bsh_\pm$, $k=1,\ldots ,n$,
we have $\tau_{n,k} w_k = \pm w_k$ implying 
$\sum_{k=1}^n \tau_{n,k} w = q w^q_f$.
Thus we get the decomposition
$\varphi^\lambda = \oplus_{q \in \bZ} \varphi^{\lambda,q}$
where
$\varphi^{\lambda,0} = \varphi$,
$\varphi^{\lambda,q}(s) = \varphi(s) - q \lb s_0  \rb$;
this makes evident that 
\[
\varphi^\lambda(s) \ = \ \varphi(s) - \lb s_0 \rb Q \, .
\]
We now return on the Weyl algebra and define $\alpha \in \Aut\mW$,
$\alpha(W(s)) \doteq  e^{-i\lb s_0 \rb} W(s)$.
We have $W^{\lambda,\omega,q}(s) = \bI_q \otimes e^{-iq \lb s_0 \rb} W^\omega(s)$, 
so that $\pi^{\lambda,\omega,q} = \bI_q \otimes \pi^\omega \alpha^q$
where $\bI_q$ is the projection onto $\mF_a^q(\bsh)$.
Using $\bI_{q-1} \psi(w) = \psi(w) \bI_q$ we find
\[
\psi(w) \, (\bI_q \otimes e^{-iq \lb s_0 \rb} W^\omega(s))
\ = \
e^{-i \lb s_0 \rb} \, (\bI_{q-1} \otimes e^{-i(q-1) \lb s_0 \rb} W^\omega(s)) \, \psi(w) 
\, ,
\]
in accord with (\ref{eq.00.01'}). 
Thus the fermionic field operators are intertwiners of the type
\[
\psi(w) \in ( \pi^{\lambda,\omega,q} , \pi^{\lambda,\omega,q-1} \alpha ) \, :
\]
besides shifting the fermionic charge $q$, they intertwine the action of
the outer automorphism $\alpha$.

\medskip 

\noindent \emph{Example 2. Yukawa potentials.}
Let $m>0$ and $\sigma(\bsx) \doteq e^{-m|\bsx|}|\bsx|^{-1}$ intended as 
the fundamental solution $(-\bs\Delta + m^2)\sigma = \delta$.
We consider the Weyl algebra $\mW$ defined by (\ref{eq.00.i1}), a GNS triple 
$( \bsh^\omega , \Omega_b , \pi^\omega )$ for $\omega \in \mS(\mW)$ and the Hilbert space 
$\mH^\omega \doteq \mF_a(\bsh) \otimes \bsh^\omega$. Assuming that $\omega$ is regular,
we get a field $\varphi^\lambda$ such that 
\[
[ \varphi^\lambda(s) \, , \, \psi(w)  ] \ = \ -\psi(\sigma \star s_0 \cdot w)
\ \ \ , \ \ \ 
[ \varrho(s) \, , \, \psi(w)  ] \ = \ -\psi(s_0w)
\, ,
\]
where $\varrho(s) \doteq \varphi^\lambda((-\bs\Delta + m^2)s)$. 
Since $\supp(\sigma) = \bR^3$, we conclude that $\psi$ and $\varphi^\lambda$
are not relatively local. 
Despite that, $\varrho$ behaves like a charge density in the sense that
it generates local gauge transformations of $\psi$ 
and hence is able to detect the electron-positron charge,
\[
[ \varrho(s) \, , \, \psi(w)  ] \ = \ - \psi(w)
\, ,
\]
having chosen $s_0$ such that $s_0 \restriction \supp(w) = 1$. 
%
%
Defining $\Omega_c \doteq \psi(w) \Omega$, $w \in \bsh_+$, 
where $\Omega \doteq \Omega_f \otimes \Omega_b$,
for $\omega$ the standard Fock state we obtain the charged state 
\begin{align}
\label{eq.state.Yukawa}
\omega_{\sigma,-1}(W(s)) \ \doteq \
\lb \Omega_c \, , W^{\lambda,\omega}(s) \Omega_c \rb \ = 
\ \ \ \ \ \ \ \ \ \ \ \ \ \ \ \ \ \ \ \ \ \ \ \ \ \ \ \  \\  
\int 
\lb w(\bsx) , 
e^{-i \int e^{-m|\bsx - \bsy|}|\bsx - \bsy|^{-1} s_0(\bsy) d^3\bsy } 
w(\bsx) \rb
\, d^3\bsx \cdot 
e^{-1/4 \| s_0 + i s_1 \|_2^2 }
\, .
\end{align}
As observed in the general discussion, $\omega_{\sigma,-1}$ is not a
local excitation of $\omega$. 
Analogously we obtain states of $\mW$ with arbitrary charge $q \in \bZ$.

\medskip 

\noindent \emph{Example 3. Coulomb gauge.}
In the present paragraph $\sigma(\bsx) \doteq (4\pi |\bsx|)^{-1}$ is 
%
%
intended as a fundamental solution
of the Poisson equation, $-\bs\Delta \sigma = \delta$.
About the topic of Coulomb gauge we follow \cite[\S 7.8.1]{Strocchi}.
In this setting, one deals with a quantum electromagnetic potential $A = ( A_0 , \bsA)$ such that
${\bf div} \bsA = 0$
and a Dirac electron field $\psi$, with equal-time commutation relations
\begin{equation}
\label{eq.CG03}
[ \bsE_h(\bsx,t) \, , \, \psi(\bsy,t) ] 
\, = \, 
- \partial_h \sigma(\bsx - \bsy) \psi(\bsy,t)
\end{equation}
where $\bsE_h \doteq \partial_h \bsA_0 - \partial_0 \bsA_h$.
%
%
Other required commutation relations are
\[
[ \psi^*(\bsx,t) , \psi(\bsy,t) ]_+
\, = \, 
\delta(\bsx - \bsy)
\]
and, for $i=1,2,3$,
\begin{equation}
\label{eq.CG04}
[ \bsA_i(\bsx,t) , \dot{\bsA_j}(\bsy,t) ] \ = \ i\delta^{tr}_{ij}(\bs x - \bs y) 
\ \ \ , \ \ \ 
\delta^{tr}_{ij}(\bs x) 
\, \doteq \, 
\delta_{ij} \delta(\bsx) + \partial_i\partial_j \sigma(\bsx) \, ,
\end{equation}
%
%
all the other (anti)commutators being trivial.
In the free case we may pick $\bsA_0 = 0$, whilst in the interacting case
$\bsA_0 = - \bs\Delta^{-1}j_0$
where $j$ is the current of the Dirac field
and $\bs\Delta^{-1}$ is defined as a pseudo-differential operator.

\medskip

We want to show how our construction allows to exhibit a fixed-time model 
in which (\ref{eq.CG03}) holds.
In the present setting we do not involve $j_0$, because there are consolidated arguments that 
its restriction at a fixed time is not well-defined \cite[\S 4.6.1]{Strocchi}.
Instead we make use of our twisted field $\varphi^\lambda$, playing the role of $A_0$ and
related --as we shall see-- to the charge density.
To this end, we consider the usual Hilbert spaces $\bsh$ and $\mF_a(\bsh)$
defining $\psi$; then we consider the symplectic spaces
$(\mS,\eta)$ and $(\mS_{tr},\eta_{tr})$,
where $\eta$ is given by (\ref{eq.00.i1}) and
\begin{equation}
\label{eq.CG06}
\eta_{tr}(\vec\bsf,\vec\bsg) \, \doteq \,  
\int ( \bsf_1 \cdot P_{tr} \bsg_0 - 
       \bsg_1 \cdot P_{tr} \bsf_0 )
\, ,
\end{equation}
for $\vec\bsf = ( \bsf_0,\bsf_1 )$, $\vec\bsg = ( \bsg_0,\bsg_1 ) 
\in \mS_{tr} \doteq \mS(\bR^3,\bR^3) \oplus \mS(\bR^3,\bR^3)$, 
with
\[
P_{tr} \bsf_{0i}(\bsx)
\ \doteq \ 
\int \delta^{tr}_{ij}(\bsx - \bsy) \, \bsf_{0j}(\bsy) \, d^3\bsy \, .
\]
Given test functions $s \in \mS(\bR^3)$ and $\bsf \in \mS(\bR^3,\bR^3)$,
one can check straightforwardly that  
$P_{tr}\bsf \in C^\infty(\bR^3,\bR^3) \cap L^2(\bR^3,\bC^3)$
{\footnote{
To check the $L^2$ condition, we note that the Fourier transform of the term
$\partial_i\partial_j \sigma(\bsx - \bsy) \bsf(\bsy)$
reads
$-\bsk_i \bsk_j |\bsk|^{-2} \hat{\bsf}(\bsk) \sim 
 c_{ij}(\alpha,\theta) \hat{\bsf}(r,\alpha,\theta)$,
having passed to spherical coordinates,
where $c_{ij}$ are suitable continuous functions on the angles $\alpha,\theta$.
Thus one uses the fact that $\hat{\bsf}$ is $L^2$.
}}, 
$\Div P_{tr}\bsf = 0$, 
$P_{tr}\bs\nabla s = 0$, 
$P_{tr} = P_{tr}^2 = P_{tr}^*$,
%
%
where the adjoint is in the sense of the scalar product of $L^2(\bR^3,\bC^3)$.
We denote the corresponding Weyl algebras by $\mW$ and $\mW_{tr}$ respectively, 
and write their spatial tensor product $\mA \doteq \mW \otimes \mW_{tr}$.
Generators of $\mA$ can be written as products 
$W(s) W_{tr}(\vec\bsf)$
of commuting unitary symbols, fulfilling separately the Weyl relations 
with respect to the above symplectic forms. 
Let now $( \bsh_0 , \Omega_0 , \pi_0 )$ and $( \bsh_{tr} , \Omega_{tr} , \pi_{tr} )$
be GNS triples for $\omega_0 \in \mS(\mW)$ and $\omega_{tr} \in \mS(\mW_{tr})$ respectively,
%
%
and $\omega \doteq \omega_0 \otimes \omega_{tr}$,
$\omega \in \mS(\mA)$, denote the product state.
We consider the Hilbert space 
$\mH^\omega \doteq \mF_a(\bsh) \otimes \bsh_0 \otimes \bsh_{tr}$
on which the operators
$\psi(w)$, $W^\omega(s)$, $W_{tr}^\omega(\vec\bsf)$
act. Given the usual twist $u_\sigma$ (\ref{eq.lem.A.1'}), we have the unitary operators
$W^{\lambda,\omega}(s) \in \mU(\mH^\omega)$ defined as in (\ref{eq.00.02}) up to tensoring by the 
identity of $\bsh_{tr}$.
Setting

\medskip 

\[
V^\lambda(\bsf) \ \doteq \ W^{\lambda,\omega}(\Div \bsf) \, W_{tr}^\omega(-\bsf,{\bs 0})
\ \ \ , \ \ \ 
\bsf \in \mS(\bR^3,\bR^3)
\, ,
\]

\medskip 

\noindent we have that $V^\lambda(\bsf)$ acts non-trivially on all the factors 
of $\mH^\omega$. We find 
\begin{equation}
\label{eq.CG.exp}
\left\{
\begin{array}{l}
[ \psi(w) \, , \, W_{tr}^\omega(\vec\bsf) ] \, = \, 0 \, , \\ \\ 
V^\lambda(\bsf) \, \psi(w) \ = \ \psi(e^{-i\sigma\star \Div \bsf}w) \, V^\lambda(\bsf) \, , \\ \\ 
V^\lambda(-\bs\nabla s_0) \, \psi(w) \ = \ \psi(e^{-is_0}w) \, V^\lambda(-\bs\nabla s_0)\, ,
\end{array}
\right.
\end{equation}

\medskip 

\noindent 
having used (\ref{eq.00.01'}) for the second relation and $-\bs\Delta\sigma = \delta$ 
for the latter.

Let now $\omega_0$ and $\omega_{tr}$ be regular. Then 
$W^\omega(s) = e^{ i ( \phi(s_0) + \dot{\phi}(s_1) ) }$,
$W_{tr}^\omega(\vec\bsf) = e^{ i ( \bsA(\bsf_0) + \dbsA(\bsf_1) ) }$
for fields
$\varphi(s)$, $\phi(s_0) \doteq \varphi(s_0,0)$, $\dot{\phi}(s_1) \doteq \varphi(0,s_1)$, 
$\bsA(\bsf_0)$, $\dbsA(\bsf_1)$, and
\begin{equation}
\label{eq.CG07'}
[ \varphi(s) \, , \, \bsA(\bsf_0) ] \, = \, 
[ \varphi(s) \, , \, \dbsA(\bsf_1) ] \, = \, 0
\, ,
\end{equation}
\begin{equation}
\label{eq.CG07}
[ \phi(s_0) \, , \, \dot\phi(s_1) ] \, = \, i \int s_0 s_1
\ \ \ , \ \ \ 
[ \bsA(\bsf_0) \, , \, \dbsA(\bsf_1) ] \, = \, i \int \bsf_0 \cdot P_{tr} \bsf_1
\, .
\end{equation}
The twist $u_\sigma$ yields as usual the field $\varphi^\lambda$
and consequently we set $\phi^\lambda(s_0) \doteq \varphi^\lambda(s_0,0)$.
We define the electric field
\[
\bsE^\lambda(\bsf) \, \doteq \, \phi^\lambda(\Div \bsf) - \dbsA(\bsf)
\ \ \ , \ \ \ 
\bsf \in \mS(\bR^3,\bR^3)
\, .
\]
With the above definitions, by applying (\ref{eq.00.06}) we find

\medskip 

\begin{equation}
\label{eq.CG08}
[ \bsE^\lambda(\bsf) , \psi(w) ] \ = \
[ \phi^\lambda(\Div\bsf) , \psi(w) ] \ = \
- \psi( \sigma \star \Div\bsf \cdot w )
\, ,
\end{equation}

\medskip

\noindent that yields the rigorous counterpart of (\ref{eq.CG03}). 
Defining the charge density field operators
$\Div \bsE^\lambda(s_0) \doteq - \bsE^\lambda( \bs\nabla s_0)$,
$s_0 \in \mS(\bR^3)$,
we find
\begin{equation}
\label{eq.CG09}
\Div \bsE^\lambda(s_0) \ = \ -\phi^\lambda(\bs\Delta s_0) + \dbsA(\bs\nabla s_0)
\, ;
\end{equation}
keeping in mind that $-\bs\Delta\sigma = \delta$, and applying (\ref{eq.00.18}), we conclude that 

\medskip 

\begin{equation}
\label{eq.CG10}
[ \Div \bsE^\lambda(s_0) \, , \, \psi(w) ] \ = \
- [ \phi^\lambda(\bs\Delta s_0) \, , \, \psi(w) ] \ = \
- \psi( s_0 w )
\, ,
\end{equation}
\begin{equation}
\label{eq.CG10'}
e^{i\Div \bsE^\lambda(s_0)} \, \psi(w) \, e^{-i\Div \bsE^\lambda(s_0)}
\ = \
\psi( e^{-is_0} w )
\, ,
\end{equation}

\medskip 

\noindent obtaining that $\Div \bsE^\lambda$ is a generator of the local gauge
transformations of $\psi$.
Note that $\Div \bsE^\lambda$ and $\psi$ are relatively local 
despite $\bsE^\lambda$ and $\psi$ are not.
If we take a bump function $s_0$ such that $s_0 \restriction \supp(w) = 1$,
then 
\[
[ \Div \bsE^\lambda(s_0) \, , \, \psi(w) ] \ = \ - \psi(w) \, ,
\]
thus, as in the previous example,
the charge density operators are able to detect the electron-positron charge.
On the other hand, $\Div \bsE^\lambda$ and $\bsA$ commute,
\begin{align}
[ \Div \bsE^\lambda(s_0) , \bsA(\bsf) ] & = \ 
[ \dbsA(\bs\nabla s_0) , \bsA(\bsf) ] \ = \ 
-i \int \bs\nabla s_0 \cdot P_{tr} \bsf \ = \ 0 \, .
%
\end{align}
%
%
%
Until now we did not made assumptions on $\omega_{tr}$, aside regularity.
We now pick the state
$\omega_{tr}(W_{tr}(\vec\bsf)) \ \doteq \ e^{-1/4 \| P_{tr}\vec\bsf \|}$,
$P_{tr}\vec\bsf \, \doteq \, P_{tr}(\bsf_0 + i\bsf_1)$;
it is induced by the Fock vector $\Omega_{tr} \in \bsh_{tr} \doteq \mF_s(\bsh^1_{tr})$, 
\[
\bsh^1_{tr} \ \doteq \ 
{\mathrm{span}} \{ P_{tr}\bsv \, , \, \bsv \in \mS(\bR^3,\bC^3) \}
\subset L^2(\bR^3,\bC^3) \, . 
\]
This yields the usual Coulomb gauge for the free field, based on 
creation and annihilation operators
$\bsa_b(P_{tr}\vec\bsf)$, $\bsa_b^*(P_{tr}\vec\bsf)$.
Since both $\bsA(\bsf_0)$ and $\dbsA(\bsf_1)$ are defined by combinations
of $\bsa_b(P_{tr}\vec\bsf)$ and $\bsa_b^*(P_{tr}\vec\bsf)$, 
and since $P_{tr}\bs\nabla s_0 \equiv 0$,
we conclude that the Coulomb condition 
$\bsA(\bs\nabla s_0) = \dbsA(\bs\nabla s_0) = 0$, $s_0 \in \mS(\bR^3)$,
holds, implying 
\[
\Div \bsE^\lambda(s_0) \ = \ -\phi^\lambda(\bs\Delta s_0) \, .
\]

\section{Momenta and Hamiltonian}
\label{sec.H}


In this section we give a sketch of the Hamiltonian of $\varphi^\lambda$
at a quite formal level
{\footnote{
Of course in general Hamiltonians should be understood as quadratic forms \cite[Vol.2, \S X.7]{RS}.
}}.
Our purpose is essentially pedagogic, in that the Hamiltonian, 
that we discuss in momentum space,
gives a hint of the behaviour that $\varphi^\lambda$ should have in a relativistic setting.
For this reason, and since at a fixed time the Dirac current 
is not available \cite[§4.6.1]{Strocchi}, we will not consider terms related to the Dirac field.
%
%
%
To fix ideas, we consider the case of the previous Example 2 so that 
$\sigma$ is the Yukawa potential for some $m > 0$; 
moreover, we pick $\omega \in \mS(\mW)$ as the Fock state.

\medskip 

As we saw in the previous sections $\varphi^\lambda$ commutes with the fermionic 
particle number and stabilizes any Hilbert space of the type (\ref{eq.00.Hilb}).
For the Hamiltonian we shall study accordingly the corresponding restrictions.
The restriction of $\varphi^\lambda$ to the zero fermion subspace $\mH^{\omega,0}$
is the untwisted field $\varphi$, thus we have no interest in it.

\paragraph{One-fermion subspace.}
Given a normalized Dirac spinor $w \in \bsh_+$, 
we consider the restriction of $\varphi^\lambda$ to the Hilbert space
(\ref{eq.00.Hilb}) that here we write
\begin{equation}
\label{eq.hilb}
\mH \ \doteq \ 
{\mathrm{span}} \, 
\{ e^{-i\sigma \star s_0} w \otimes W(s)\Omega_b \, , \, 
s \in \mS \} 
\ \subset \bsh_+ \otimes \mF_s(\bsh_0)
\, .    
\end{equation}
%
%
%
%
Now, the field $\phi^\lambda(s) \doteq \varphi^\lambda (s,0)$,
$s \in \mS(\bR^3)$, is defined by the creation and annihilation operators (\ref{eq.tbf02})
that in $\mH$ have the form
\begin{equation}
\label{eq.tbf02'}
\bsa_{b,\sigma}^\circ(s) \, \doteq \, 
\bI_f \otimes \bsa_b^\circ(s) + \frac{1}{\sqrt{2}} \bmu_{\sigma,s} \otimes \bI_b  
\, ,
\end{equation}
where $\bmu_{\sigma,s}$ is the one-particle operator acting on $\bsh_+$.
The corresponding pointwise-defined objects are
\begin{equation}
\label{eq.tbf02'''}
\bsa_{b,\sigma}^\circ(\bsx) \, \doteq \, 
\bI_f \otimes \bsa_b^\circ(\bsx) + \frac{ \bmu_\sigma(\bsx) }{\sqrt{2}}  \otimes \bI_b 
\, ,
\end{equation}
%
%
%
where
$\bmu_\sigma(\bsx) w(\bsy) \doteq - \sigma(\bsy - \bsx)w(\bsy)$.
%
%
Passing to momentum space, we have that the Fourier transform of $\sigma$ is
$\hat\sigma(\bsk) = 4\pi\varpi(\bsk)^{-2}$, $\varpi(\bsk) \doteq \sqrt{|\bsk|^2 + m^2}$. 
%
%
%
%
%
The creation and annihilation operators are
\begin{equation} 
\label{eq.tbf02'k} 
\bsa_{b,\sigma}^\circ(\hat{s}) \, \doteq \,  
\bI_f \otimes \bsa_b^\circ(\hat{s}) + \frac{1}{\sqrt{2}} \hat{\bmu}_{\sigma,s} \otimes \bI_b 
\ \ , \ \ 
\hat{\bmu}_{\sigma,s}\hat{w}  \, \doteq \, 
\widehat{\bmu_{\sigma,s}w}  \, = \, 
-(\hat\sigma \hat{s}_0) \star \hat{w}
\, ,
\end{equation} 
with sharp momentum operators
\begin{equation} 
\label{eq.tbf02''} 
\bsa_{b,\sigma}^\circ(\bsk) \, \doteq \,  
\bI_f \otimes \bsa_b^\circ(\bsk) + 
\frac{1}{\sqrt{2}} \hat{\bmu}_\sigma^\circ(\bsk) \otimes \bI_b 
\, ,  
\end{equation} 
\begin{equation} 
\label{eq.tbf02''k} 
\hat{\bmu}_\sigma^*(\bsk)  \wa{w}(\bsk') \ \doteq \ 
- \hat{\sigma}(\bsk) \hat{w}(\bsk' - \bsk) \ = \ 
- \frac{4\pi}{\varpi(\bsk)^2} \, \hat{w}(\bsk' - \bsk) \, .
\end{equation}

\noindent Note that, despite $\hat{\bmu}_{\sigma,s}$ is selfadjoint,
$\hat{\bmu}_\sigma^*(\bsk)$ has adjoint 
$\hat{\bmu}_\sigma(\bsk) = \hat{\bmu}_\sigma^*(-\bsk)$
due to the translation action on $\hat{w}$;
thus we made an arbitrary choice when we defined $\hat{\bmu}_\sigma^*(\bsk)$ 
in (\ref{eq.tbf02''k}) rather than $\hat{\bmu}_\sigma(\bsk)$,
dictated by the wish to obtain, in the following (\ref{eq.Hlambda'}), 
the term $\hat{\bmu}_\sigma^*(\bsk) \bsa_b(\bsk) + \hat{\bmu}_\sigma(\bsk) \bsa^*_b(\bsk)$.
The operators $\hat{\bmu}_\sigma^\circ(\bsk)$ are bounded and, by the estimate
\[
\int | \hat{\bmu}_\sigma(\bsk)  \wa{w}(\bsk') - \hat{\bmu}_\sigma(\bsh) \wa{w}(\bsk') |^2 \, d^3\bsk'
\ \leq \ 
\frac{1}{m^4} \int \left| 
\wa{w}(\bsk + \bsk') - \frac{|\bsk|^2 + m^2}{|\bsh|^2 + m^2} \wa{w}(\bsh + \bsk') 
\right|^2 d^3\bsk' \, ,
\]
depend continuously on $\bsk$ in the strong topology
(provided that $\bsh,\bsk$ vary in an arbitrary compact set). 
Moreover, since they act exclusively on the fermionic factor of $\mH$, 
they commute with the bosonic operators $\bsa_b^\circ(\bsk)$.
Thanks to $\hat{\bmu}_\sigma^\circ(\bsk)$, we have that $\bsa_{b,\sigma}^\circ(\bsk)$ performs a shift 
of the momenta of the fermionic states, besides creating and annihilating bosonic states.
The Hamiltonian 
\begin{equation}
\label{eq.Hlambda}
H^\lambda \, \doteq \, 
\int \varpi(\bsk) \, \bsa^*_{+\sigma}(\bsk) \bsa_{b,\sigma}(\bsk) \, d^3\bsk
\end{equation}
is written explicitly 
%
%
\begin{equation}
\label{eq.Hlambda'}
\int
\varpi(\bsk)
\left( 
\bsa^*_b(\bsk) \bsa_b(\bsk) -
\frac{1}{\sqrt{2}} 
\left( \hat{\bmu}_\sigma^*(\bsk) \bsa_b(\bsk) +
       \hat{\bmu}_\sigma(\bsk) \bsa^*_b(\bsk) 
\right) + 
\frac{1}{2} \hat{\bmu}_\sigma^*(\bsk) \hat{\bmu}_\sigma(\bsk)
\right) \, d^3\bsk \, .
\end{equation}
%
%
%
%
%
%
The previous expression shows that the Hamiltonian density splits into three different terms, say 
$\efH^0$, $\efH^{\bmu\varphi}$ and $\efH^{\bmu}$. 
The first term $\efH^0$ is the free Hamiltonian density, whilst $\efH^{\bmu\varphi}$
involves annihilation and creations of bosons besides affecting momenta of fermionic wave functions.
We have
\[
\left\{
\begin{array}{l}
[ \efH^0(\bsk') \, , \, \efH^{\bmu\varphi}(\bsk) ] = 
2^{-1/2}\varpi(\bsk) \left(
\hat{\bmu}_\sigma(\bsk) \bsa^*_b(\bsk) -
\hat{\bmu}_\sigma^*(\bsk) \bsa_b(\bsk) \right) \delta(\bsk' - \bsk) \, , 
\\ \\ 
\left[ \efH^\bmu \, , \, \efH^0 \right] \, = \, 
\left[ \efH^\bmu \, , \, \efH^{\bmu\varphi} \right] \, = \, 0 \, ,
\end{array}
\right.
\]
thus the boson number is not preserved and the terms $\efH^{\bmu\varphi}$, $\efH^{\bmu}$
forming the interacting part of the Hamiltonian commute.
Note that we may write 
\[
\varpi(\bsk) \hat{\bmu}_\sigma^*(\bsk) \, = \, - \frac{4\pi}{\varpi(\bsk)} U_\bsk 
\, ,
\]
where $U_\bsk \in \mU(\bsh_+) \otimes \bI_b$ are the translations by $\bsk$.
Up to these unitaries, $H^\lambda$ is analogous to a Van Hove Hamiltonian 
of the second kind in the sense of \cite[\S 1]{Der03}. Since
%
%
%
\[
\int \frac{1}{\varpi(\bsk)^4} d^3\bsk \, < \, \infty 
\ \ , \ \ 
\int \frac{1}{\varpi(\bsk)^3} d^3\bsk \, = \, 
\int \frac{1}{\varpi(\bsk)^2} d^3\bsk \, = \, 
\infty 
\, ,
\]
we may give a meaning to $H^\lambda$ as a selfadjoint operator
following a route analogous to \cite[\S 1.1, (3)]{Der03},
which requires an ultraviolet renormalization.
We do not deepen this point since it is beyond the scope 
of this section. 

\medskip 

Integrating the Hamiltonian density with the ultraviolet cutoff 
$g(\bsk) d^3\bsk$, $g \in \mS(\bR^{3,*})$,
we get operators $\efH^{\bmu\varphi}_g$, $\efH^\bmu_g$. 
%
%
%
%
In particular,
\[
\efH^\bmu_g  \, = \, 
\int \frac{8\pi^2}{\varpi(\bsk)^3} g(\bsk) \, d^3\bsk 
\]
is a multiple of the identity. 
Defining $H^\lambda_g \doteq H^0 + \efH^{\bmu\varphi}_g + \efH^\bmu_g$,
where $H^0$ is the free Hamiltonian,
it is instructive to compute the state
\begin{equation}
\label{eq.H1}
H^\lambda_g ( \hat{w} \otimes \Omega_b ) \, = \, 
%
%
  \int \frac{4\pi g(\bsk)}{\varpi(\bsk)} \, \hat{w}_\bsk \otimes |\bsk\rb \, d^3\bsk
+ \int \frac{8\pi^2}{\varpi(\bsk)^3} g(\bsk) \, d^3\bsk \cdot \hat{w} \otimes \Omega_b \, ,
\end{equation}
where $| \bsk \rb \doteq \bsa^*_b(\bsk)\Omega_b$ and 
$\hat{w}_\bsk(\bsk') \doteq \hat{w}(\bsk' + \bsk)$. 
The first term at r.h.s. is given by $\efH^{\bmu\varphi}_g( \hat{w} \otimes \Omega_b )$ 
and is a superposition of the improper states $\hat{w}_\bsk \otimes |\bsk\rb$.
The translation by $\bsk$ in $\hat{w}_\bsk$ indicates that the momentum
of $\hat{w}$ has been affected by the creation of the improper bosonic state $|\bsk\rb$.
To better illustrate this point, let us take $g$ a smooth approximant of the Dirac delta 
on $\bsk_\bullet \in \bR^3$ (with small compact support around $\bsk_\bullet$) 
and $\hat{w}$ with support in a small neighbourhood of $\bsk_e \in \bR^3$:
to emphasize the almost sharp momentum, we write 
$| \bsk_\bullet \rb_\infty = \bsa^*_b(g)\Omega_b$ and $\hat{w} = | \bsk_e \rb_\infty$.
Now, the support of $\hat{w}_{\bsk_\bullet}$ lies in a small neighbourhood of $\bsk_e - \bsk_\bullet$,
thus we write $\hat{w}_{\bsk_\bullet} = | \bsk_e - \bsk_\bullet \rb_\infty$. 
Treating $g$ as a Dirac delta we find the approximation
\begin{equation}
\label{eq.H2}
H^\lambda_g \, ( | \bsk_e \rb_\infty \otimes \Omega_b) 
\, \simeq \, 
\frac{4\pi}{\varpi(\bsk_\bullet)} \cdot | \bsk_e - \bsk_\bullet \rb_\infty \otimes | \bsk_\bullet \rb_\infty + 
 \frac{8 \pi^2}{\varpi(\bsk_\bullet)^3} \cdot | \bsk_e \rb_\infty \otimes \Omega_b \, ,
\end{equation}
making evident that $| \bsk_e \rb_\infty$ has lost a momentum $\bsk_\bullet$ gained by
the bosonic state $| \bsk_\bullet \rb_\infty$.
The second term $\efH^\bmu_g( \hat{w} \otimes \Omega_b )$ --being defined by a c-number--
mimics the effect that would be obtained by introducing an external potential interacting with the
untwisted field $\phi$.
%
%
Note that $H^\lambda ( \hat{w} \otimes \Omega_b )$ is obtained in the limit $g \to 1$, 
in which both $\efH^{\bmu\varphi}_g$ and $\efH^\bmu_g$ diverge.
This is not surprising as, in particular, the vector valued distribution
$\bsk \mapsto \hat{w}(\bsk' - \bsk) \otimes |\bsk\rb$
must be smeared by a test function to give a meaningful vector. 

\paragraph{Electron-positron subspace.}
Finally we consider the state of $\mW$ induced by the vector $w \doteq w_1 \otimes_a \bar{w}_2$,
$w_1 \in \bsh_+$, $\bar{w}_2 \in \bsh_-$, standing in the two-fermion space and
having zero charge. The GNS Hilbert space is
\[
\mH \ \doteq \ 
{\mathrm{span}} \, 
\{ \iGamma_a(u_{\sigma,s}) w \otimes W(s)\Omega_b \, , \, s \in \mS \} 
\ \subset ( \bsh_+ \otimes_a \bsh_- ) \otimes \mF_s(\bsh_0)
\, .
\]
The creation and annihilation operators are given by (\ref{eq.tbf02}) where
\begin{align*}
\bmu_\sigma^a(s) w(\bsy_1 , \bar\bsy_2) & = \ 
\left( \bmu_{\sigma,s}w_1 \otimes_a w_2 + w_1 \otimes_a \bmu_{\sigma,s} w_2 \right) (\bsy_1 , \bar\bsy_2)
\\ & = \ 
\int 
( \sigma( \bar\bsy_2 - \bsx ) - \sigma( \bsy_1 - \bsx )) s_0(\bsx) \, d^3\bsx 
\cdot w_{sym}(\bsy_1,\bar\bsy_2)
\, , 
\end{align*}
\[
w_{sym}(\bsy_1 , \bar\bsy_2) \ \doteq \
w_1(\bsy_1) w_2(\bar\bsy_2) + w_1(\bar\bsy_2) w_2(\bsy_1) 
\]
(note that $w_{sym}$ is not in the fermionic Fock state, and becomes antisymmetric
only after multiplying by the $\sigma$-factors).
The pointwise-defined operators are as in (\ref{eq.tbf02'''}) with
\[
\bmu_\sigma^a(\bsx)w(\bsy_1 , \bar\bsy_2) \ = \ 
( \sigma( \bar\bsy_2 - \bsx ) - \sigma( \bsy_1 - \bsx ) ) \, w_{sym}(\bsy_1,\bar\bsy_2)
\]
in place of $\bmu_\sigma(\bsx)$. Analogously, the Fourier transforms are 
as in (\ref{eq.tbf02'k}) with 
\[
\hat{\bmu}_{\sigma,s}^- \hat{w} \ \doteq \ 
( \hat{\sigma}_2 \hat{s}_0 - \hat{\sigma}_1 \hat{s}_0 ) \star \hat{w}_{sym}
\ \ \ , \ \ \ 
\hat{\sigma}_i(\bsk_i) \doteq \int \sigma(\bsy_i) e^{-i \bsk_i \bsy_i} d^3\bsy_i
\, ,
\]
and (non-selfadjoint) operators in momentum space defined by the terms
\begin{align*}
\hat{\bmu}_\sigma^-(\bsk_1,\bar\bsk_2) \hat{w}(\bsk'_1 , \bar\bsk'_2) 
\ = 
\ \ \ \ \ \ \ \ \ \ \ \ \ \ \ \ \ \ \ \ \ \ \ \ \ \ \ \ \ \ \ \ \ \ \ \ \ \ \ 
\\   
\frac{4\pi}{\varpi(\bar\bsk_2)^2} 
\left( \hat{w}_1(\bsk'_1) \hat{w}_2(\bar{\bsk}'_2 + \bar{\bsk}_2 ) + 
\hat{w}_1(\bar{\bsk}'_2 + \bar{\bsk}_2 ) \hat{w}_2(\bsk'_1) \right)
+ \\ -
\frac{4\pi}{\varpi(\bsk_1)^2}
\left( \hat{w}_1(\bsk'_1 + \bsk_1 ) \hat{w}_2(\bar\bsk'_2) + 
\hat{w}_1(\bar\bsk'_2) \hat{w}_2(\bsk'_1 + \bsk_1 ) \right) .
\end{align*}
Again, a shift of momenta of the fermionic states appears, 
alternatively on the electron and positron components.
The rest of the discussion is analogous to the one-fermion case,
thus we do not go into details.

%

\section{Conclusions}
\label{sec.concl}

In the present paper we presented a mathematically rigorous construction
of fixed-time models in which the commutation relations (\ref{eq.intro.2}) hold. 
The input is given by a free Dirac field $\psi$, a free bosonic field $\varphi$,
and a distribution $\sigma$ encoding the obstruction for them to commute. 
The output includes a new bosonic field $\varphi^\lambda$ 
acting non-trivially on the fermion subspace of the model and,
when the support of $\sigma$ has non-empty interior, not relatively local to $\psi$.

\medskip 

The so-obtained models are able to reproduce well-known commutation relations,
as those of the interacting Coulomb gauge at equal times. 
An interesting point is that, provided that $\sigma$ is fundamental solution
of a differential operator $P_\partial$, one automatically gets that 
$P_\partial \varphi^\lambda$ is local to $\psi$ and, more precisely, 
generator of the local gauge transformations of $\psi$.
As a consequence we have that $P_\partial \varphi^\lambda$ comes related to the
electron-positron charge (as a charge density), 
and that fermionic vectors induce non-localized states of the 
Weyl algebra associated with $\varphi$. Yet, restrictions of such states
to the Weyl subalgebra generated by test functions of the type 
$P_\partial s$, $s \in \mS$, are localized on the supports of the 
corresponding fermionic test functions.
The Hamiltonian of the bosonic field exhibits non-trivial interaction 
terms, that perform creations and annihilation of bosonic states and  
affect momenta of fermionic states.
It is this last property that marks a difference between our models and 
those obtained by adding an external potential to the bosonic field: 
in the latter case, even if the bosonic particle number is not preserved
there are no modifications of fermionic states.
This holds in particular for the models \cite{BCRV19,BCRV22,BCRV23},
in which the charged (classical) field is not affected by 
the bosonic (electromagnetic) field but induces localized automorphisms
on the associated C*-algebra. In this sense, our model Example 3 is
complementary to those of the above-mentioned references
{\footnote{The author is indebted to D. Buchholz for useful comments on this point.}}.
%

\medskip 

These properties make desirable to investigate the relativistic
counterparts of our class of models. 
The strategy is to consider twisting operators on the space of elementary
excitations of a relativistic charged field, rather than work with the expression 
$\phi^\lambda_{rel}(\bsx,t) \doteq e^{itH^\lambda} \phi^\lambda(\bsx) e^{-itH^\lambda}$.
Once that the relativistic models are constructed
the first point to discuss is locality, because physically interesting 
distributions (Yukawa and Coulomb potentials), as well as their relativistic counterpart, 
have a non-trivial support and hence produce non-local models. 
We will approach this problem by studying locality on the algebras generated 
by observable quantities. 
A further aspect to be explored is the analysis of the shift induced by $\varphi^\lambda$
on the supports of charged states in momentum space, that acquires a new significance
in a relativistic setting.
%
%
These topics are the subject of a work in progress \cite{Vas3}.



\newpage  
{\small
\begin{thebibliography}{99}






\bibitem{BLOT}
N.N.Bogolubov, A.A.Logunov, A.I.Oksak, I.T.Todorov:
General Principles of Quantum Field Theory.
Kluwer Academic Publishers, Dordrecht-Boston-London, 1987.


\bibitem{Bon70}
P.J.M. Bongaarts:
The electron-positron field coupled to external electromagnetic potentials
as an elementary C*-algebra theory.
Ann. Phys. 56 (1970) 108-139.

\bibitem{BR}
O. Bratteli, D.W. Robinson:
Operator Algebras and Quantum Statistical Mechanics, vol.1-2.
Springer Verlag, Berlin Heidelberg, 1997.




\bibitem{BCRV19}
D. Buchholz, F. Ciolli, G. Ruzzi, E. Vasselli:
On string-localized potentials and gauge fields. 
Lett. Math. Phys. 109 (2019) 2601-2610.


\bibitem{BCRV22}
D. Buchholz, F. Ciolli, G. Ruzzi, E. Vasselli:
The universal algebra of the electromagnetic field III. 
Static charges and emergence of gauge fields. 
Lett. Math. Phys. 112 (2022) 27.  

\bibitem{BCRV23}
D. Buchholz, F. Ciolli, G. Ruzzi, E. Vasselli:
Gauss’s law, the manifestation of gauge fields,
and their impact on local observables.
In: Cinto, A., Michelangeli, A. (eds) “Trails in Modern
Theoretical and Mathematical Physics: A Volume in Tribute to Giovanni Morchio” Springer ISBN-13 978-3031449871 (2023).
%




\bibitem{Con}
F. Constantinescu:
Distributions and their Applications in Physics.
Pergamon Press, 1980.
















\bibitem{Der03}
J. Derez\'inski:
Van Hove Hamiltonians. Exactly Solvable Models of the Infrared and Ultraviolet Problem.
Ann. Henri Poincar\'e 4 (2003) 713–738.
%

\bibitem{Her96}
A. Herdegen:
Asymptotic algebra for charged particles and radiation.
J. Math. Phys. 37 (1996) 100-120.


\bibitem{Her98}
A. Herdegen:
Semidirect product of CCR and CAR algebras and asymptotic states in quantum electrodynamics.
J. Math. Phys. 39 (1998) 1788-1817.


\bibitem{RS}
M. Reed, B. Simon:
Methods of modern mathematical physics.
Vol.I: Functional Analysis.
Vol.II: Fourier Analysis, Self-Adjointness.
Vol.III: Scattering Theory.
Academic Press, 1980.






\bibitem{Strocchi}
F. Strocchi:
An introduction to non-perturbative foundations of Quantum Field Theory.
International Series of Monographs on Physics 158.
Oxford Science Publications, 2013.


\bibitem{Vas1}
E. Vasselli:
Twisted tensor products of field algebras.
J. Noncommut. Geom. 
\href{https://doi.org/10.4171/jncg/594}{DOI 10.4171/JNCG/594}

\bibitem{Vas3}
E. Vasselli:
Twisting factors for relativistic quantum fields.
In preparation.


\end{thebibliography}
}
\end{document}